\documentclass{lmcs} 


\usepackage{hyperref}
\usepackage{proof,amsmath,amssymb,stmaryrd,ascmac}

\theoremstyle{plain} 

\newcommand{\m}{\mathbf} 
\newcommand{\mc}{\mathcal}

\newcommand{\sep}{\; | \;} 
\newcommand{\bnf}{\; ::= \;}
\newcommand{\Def}{\; := \;}
\newcommand{\Iff}{\quad \Longleftrightarrow \quad}
\newcommand{\Implies}{\quad \Longrightarrow \quad}
\newcommand{\fsubseteq}{\subseteq_{\mathsf{fin}}}
\newcommand{\fwp}{\wp_{\mathsf{fin}}}
\newcommand{\ldd}{{\bbslash}}

\newcommand{\GLC}{\mathbf{G}^1\mathbf{LC}}
\newcommand{\LI}{\mathbf{LI}}
\newcommand{\HA}{\mathbf{HA}}
\newcommand{\PA}{\mathbf{PA}}
\newcommand{\PRA}{\mathbf{PRA}}
\newcommand{\ZT}{\mathbf{PA2}}
\newcommand{\HAT}{\mathbf{HA2}}
\newcommand{\PICA}{\Pi^1_1\mbox{-{\bf CA}}_0}
\newcommand{\ID}{\mathbf{ID}}
\newcommand{\IS}{\mathbf{I\Sigma}_1}
\newcommand{\LIT}{\mathbf{LI2}}
\newcommand{\LIP}{\mathbf{LIP}}
\newcommand{\LIO}{\mathbf{LI}\Omega}

\newcommand{\Var}{\mathsf{Var}}
\newcommand{\VAR}{\mathsf{VAR}}
\newcommand{\Tm}{\mathsf{Tm}}
\newcommand{\Fm}{\mathsf{Fm}}
\newcommand{\FM}{\mathsf{FM}}
\newcommand{\FMP}{\mathsf{FMP}}
\newcommand{\ABS}{\mathsf{ABS}}
\newcommand{\SEQ}{\mathsf{SEQ}}
\newcommand{\Fv}{\mathsf{Fv}}
\newcommand{\FV}{\mathsf{FV}}

\newcommand{\Lfp}{\boldsymbol{lfp}}

\newcommand{\Nat}{\mathbb{N}}

\newcommand{\Nn}{\boldsymbol{N}}

\newcommand{\Fix}{\boldsymbol{I}}

\newcommand{\Gal}{\mathcal{G}}
\newcommand{\Rel}{\mathrel{R}}

\newcommand{\Left}{\;\mathsf{left}}
\newcommand{\Right}{\;\mathsf{right}}
\newcommand{\Cut}{\mathsf{cut}}

\newcommand{\Rcf}{\Rightarrow^{cf}}

\newcommand{\Code}[1]{\ulcorner #1 \urcorner}
\newcommand{\Rk}{\mathsf{rank}}
\newcommand{\LLIO}{\boldsymbol{LI}\Omega}
\newcommand{\LLIP}{\boldsymbol{LIP}}
\newcommand{\LLI}{\boldsymbol{LI}}
\newcommand{\Cl}{\boldsymbol{\gamma}}
\newcommand{\II}{\boldsymbol{Ip}}

\newcommand{\Sf}{\mathsf{Sf}}
\newcommand{\TTm}{\boldsymbol{Tm}}
\newcommand{\SSeq}{\boldsymbol{SEQ}}
\newcommand{\FFset}{\boldsymbol{Fset}}
\newcommand{\IF}{\boldsymbol{Ip}_{F}}

\begin{document}

\title[MacNeille completion and Buchholz' Omega rule]{MacNeille completion and Buchholz' Omega rule for parameter-free second order logics}

\author[K.~Terui]{Kazushige Terui}	
\address{Research Institute for Mathematical Sciences, Kyoto University,
Kitashirakawa-Oiwakecho, Sakyoku, Kyoto 606-8502, Japan.}	
\email{terui@kurims.kyoto-u.ac.jp}  
\thanks{Supported by KAKENHI 25330013.}





\begin{abstract}
Buchholz' \emph{$\Omega$-rule} is a way to give a syntactic, possibly 
ordinal-free proof of cut elimination for various subsystems of
second order arithmetic. 
Our goal is to understand it from an algebraic point of view. 
Among many proofs of cut elimination for higher order logics, Maehara and
Okada's algebraic proofs are of particular interest, since the essence of 
their arguments can be algebraically described as the 
(Dedekind-)\emph{MacNeille completion} 
together with Girard's reducibility candidates. Interestingly, 
it turns out that the $\Omega$-rule, formulated as a rule of logical inference, 
finds its algebraic foundation in the MacNeille completion. 

In this paper, we consider 
the \emph{parameter-free} fragments $\{\LIP_n\}_{n<\omega}$
of the second order intuitionistic logic, that correspond to 
the arithmetical theories $\{\ID_n\}_{n<\omega}$
of iterated inductive definitions up to $\omega$. 
In this setting, we observe 
a formal connection between the $\Omega$-rule and the MacNeille completion,
that leads to a way of interpreting second order quantifiers 
in a first order way in Heyting-valued semantics, called the 
\emph{$\Omega$-interpretation}. Based on this,
we give an algebraic proof of 
cut elimination for $\LIP_n$ for every $n<\omega$ 
that can be locally formalized in $\ID_n$.
As a consequence, we obtain an equivalence between
the cut elimination for $\LIP_n$ and the 1-consistency of $\ID_n$ 
that holds in a weak theory of arithmetic.
\end{abstract}

\maketitle


\section*{Introduction}

This paper is concerned with cut elimination for subsystems of 
second order logics. It is of course very well known that 
the full second order classical/intuitionistic logics admit 
cut elimination. Then why are we interested in their subsystems?
A primary reason is that proving cut elimination for a subsystem 
is often very hard if one is sensitive to the metatheory within which
(s)he works. This is witnessed by the vast literature
in the traditional proof theory. In fact, proof theorists are not 
just interested 
in proving cut elimination itself, but in identifying
a \emph{characteristic principle} $P$ 
(e.g.\ ordinals, 
combinatorial principles and inductive definitions)
for each system of logic, arithmetic and set theory,
by proving cut elimination
within a weak metatheory (e.g.\ $\PRA$, $\IS$ or
$\mathbf{RCA}_0$) extended by $P$.
Our long-span goal is to understand
those hard proofs and results from an algebraic perspective.

\subsection*{Various proofs of cut elimination}
One can distinguish several types of cut elimination proofs 
for higher order logics/arithmetic: (i) syntactic proofs by ordinal assignment
(e.g.\ Gentzen's consistency proof for $\PA$),
(ii) syntactic but ordinal-free proofs, (iii) semantic proofs
based on Sch\"{u}tte's semivaluation or its variants (e.g.\ \cite{Tait66}),
(iv) algebraic proofs based on completions (the list is not intended 
to be exhaustive).
Historically (i) and (iii) precede (ii) and (iv), but 
(i) is not easy to follow up due to the heavy use of ordinal notations,
while (iii) is not completely satisfactory for computer scientists
since it involves
\emph{reductio ad absurdum} and weak K\"{o}nig's lemma,
that would destruct the proof structure: the output cut-free proof may 
have nothing to do with the input proof.
Hence we address (ii) and (iv) in this paper.

For (ii), a very useful and versatile technique is Buchholz' 
\emph{$\Omega$-rule}. Although 
introduced in the context of ordinal analysis
\cite{Buchholz81a}, the technique itself can be understood 
without recourse to any ordinals 
\cite{Buchholz01,Akiyoshi_Mints11,Akiyoshi17}.
Still, the $\Omega$-rule is notoriously complicated and is 
hard to grasp its meaning at a glance. Even its semantic soundness is 
not clear at all. 
While Buchholz gives an account based on the BHK interpretation
\cite{Buchholz81a},
we will try to give an algebraic account in this paper.

For (iv), there is a very conspicuous algebraic proof of
cut elimination for higher order logics which may be primarily ascribed to 
Maehara \cite{Maehara91} and Okada \cite{Okada96,Okada02}.
In contrast to (iii), these algebraic proofs 
are fully constructive; no use of reductio ad absurdum or any 
nondeterministic principle. More importantly, it extends 
to proofs of normalization for proof nets and typed lambda calculi 
\cite{Okada99}. While their arguments can be described in various
dialects (e.g.\ phase semantics of linear logic),
apparently most neutral and 
most widely accepted would be to speak in terms of algebraic completions:
the essence of their arguments can be described 
as the (\emph{Dedekind}-)\emph{MacNeille
completion} together with Girard's reducibility candidates, 
as we will explain in Appendix \ref{a-okada}.

\subsection*{Contents of this paper}
Having a syntactic technique on one hand and an algebraic methodology on
the other, it is natural to ask the relationship between them.
To make things concrete, we consider
the \emph{parameter-free} fragments $\{\LIP_n\}_{n<\omega}$
of the standard sequent calculus $\LIT$ for 
the second order intuitionistic logic. These fragments altogether 
constitite an intuitionistic counterpart of 
the classical sequent calculus studied in \cite{Takeuti58}.
Although we primarily work on the intuitionistic basis,
all results in this paper (except Proposition \ref{p-counter2})
carry over to the classical logic too.

As we will see, cut elimination based on the $\Omega$-rule technique works for 
$\LIP_n$ for every $n<\omega$. 
Moreover, it turns out to be intimately related to
the MacNeille completion in that the $\Omega$-rule in our setting is 
not sound in Heyting-valued semantics in general,
but is sound when the underlying algebra is 
the MacNeille completion of the Lindenbaum algebra.
This observation leads to a curious way of interpreting second order formulas 
in a first order way, that we call the \emph{$\Omega$-interpretation}. 
The basic idea already appears 
in Altenkirch and Coquand \cite{Altenkirch_Coquand01}, but 
ours is better founded and accommodates the existential quantifier too.

The $\Omega$-interpretation in conjunction with the MacNeille completion
gives rise to an algebraic proof of (partial) cut elimination for $\LIP_n$,
that is comparable with 
Aehlig's result \cite{Aehlig05} for
the parameter-free, negative fragments
of second order Heyting arithmetic.
The $\Omega$-interpretation is essentially first order. In particular, it 
does not employ the reducibility candidates.  Hence it is 
``locally'' formalizable in (the intuitionistic version of) $\ID_n$. 
This leads to a correspondence between $\ID$-theories in arithmetic 
and parameter-free logics, that we call the \emph{Takeuti correspondence}:
the cut elimination for $\LIP_n$ is equivalent to 
the 1-consistency of $\ID_n$.

\subsection*{Outline}
The rest of this paper is organized as follows.
In Section \ref{s-macneille}
 we recall some basics of the MacNeille completion.
In Section \ref{s-takeuti}
 we review the theories of iterated inductive definitions up to $\omega$,
introduce the parameter-free systems $\LIP_n$ ($n < \omega$), 
and prove one direction of the Takeuti correspondence between 
$\LIP_n$ and $\ID_n$.
In Section \ref{s-omega}
 we introduce the $\Omega$-rules
in our logical setting and explain how they work for 
$\LIP_n$ by giving a syntactic proof of cut elimination. 
In Section \ref{s-omegamacneille}, 
we turn to the algebraic side and establish a connection between
the $\Omega$-rule and 
the MacNeille completion, that leads to the concept of $\Omega$-interpretation.
In Section \ref{s-algebraic}, 
we given an algebraic proof of (partial) cut elimination for
$\LIP_n$ based on the $\Omega$-interpretation. 
In Section \ref{s-formalize}, we sketch a (local) formalization 
of our algebraic argument, that establishes the Takeuti correspondence
between $\LIP_n$ and $\ID_n$ for every $n<\omega$.

\begin{rem}
One often distinguishes \emph{cut elimination} from \emph{cut admissibility}
(or \emph{cut eliminability}). While the former gives a concrete procedure,
the latter only ensures the existence of a cut-free derivation.
Although our algebraic argument will 
only establish cut admissibility, we prefer to use the word
cut elimination. A reason is that the statement of cut admissibility
is $\Pi^0_2$, so a concrete procedure can be extracted from its proof
(especially noting that our proof is fully constructive).
Of course, this does not ensure that the procedure respects
proof equivalence in any sense. 
Hence we do not make any formal claim on this point.
\end{rem}

\section{MacNeille completion}\label{s-macneille}

Let $\m A = \langle A, \wedge, \vee\rangle$ be  a lattice. 
A \emph{completion} of $\m A$ is an embedding
$e: \m A \longrightarrow \m B$ into a complete lattice 
$\m B = \langle B, \wedge, \vee\rangle$.
We often assume that $e$ is an inclusion map so that
$\m A$ is a subalgebra of $\m B$ (notation: $\m A \subseteq \m B$).

Here are two examples.
\begin{itemize}
\item Let $[0,1]_{\mathbb{Q}} := [0,1]\cap \mathbb{Q}$ 
be the chain of rational numbers in the unit interval 
(seen as a lattice).
Then it admits an obvious completion 
$[0,1]_{\mathbb{Q}} \subseteq [0,1]$.
\item Let $\m A$ be a Boolean algebra. Then it also admits a completion
$e: \m A \longrightarrow \m A^\sigma$, 
where $\m A^\sigma := \langle \wp(\mathsf{uf}(\m A)), 
\cap, \cup, -, A, \emptyset\rangle$,
the powerset algebra on the set of ultrafilters of $\m A$,
and $e(a) := \{u \in \mathsf{uf}(\m A) : a \in u\}$.
\end{itemize}

A completion $\m A \subseteq \m B$ is \emph{$\bigvee$-dense}
if $x = \bigvee \{a \in A: a\leq x\}$ holds for every $x\in B$.
It is \emph{$\bigwedge$-dense} if 
$x = \bigwedge \{a \in A: x \leq a\}$.
A $\bigvee$-dense and $\bigwedge$-dense completion is called 
a \emph{MacNeille completion} (or a \emph{Dedekind-Macneille completion}). 
This means that 
any $\m B$-element can be approximated from above and below
by $\m A$-elements. The following is a classical result
\cite{Banaschewski56,Schmidt56}.

\begin{thm}
Every lattice $\m A$ has a MacNeille completion unique up to isomorphism.
Any MacNeille completion is \emph{regular}, 
that is, preserves all joins and meets that already exist in $\m A$.
\end{thm}

Coming back to the previous examples:
\begin{itemize}
\item $[0,1]_{\mathbb{Q}} \subseteq [0,1]$ is MacNeille, since
$
x = \inf \{ a \in \mathbb{Q} : x \leq a\} = \sup \{a \in \mathbb{Q} : a \leq x\}$
for any $x \in [0,1]$. It is regular since if $q = \lim_{n \rightarrow \infty}
q_n$ holds in $\mathbb{Q}$, then it holds in $\mathbb{R}$ too. 
\item $e: \m A \longrightarrow \m A^\sigma$ is not regular
when $\m A$ is an infinite Boolean algebra. In fact, 
the Stone space 
$\mathsf{uf}(\m A)$ is compact, so collapses any infinite covering of 
a closed set into a finite one. 
It is actually a \emph{canonical extension},
that has been extensively studied in ordered algebra and 
modal logic \cite{JT51,GJ04,GH01}.
\end{itemize}

MacNeille completions behave better than
canonical extensions in the preservation of existing limits,
but the price to pay is the loss of generality.
Let $\mathcal{DL}$ ($\mathcal{HA}$,
$\mathcal{BA}$, resp.) be the variety of distributive lattices
(Heyting algebras, Boolean algebras, resp.).

\begin{thm}\label{t-limit}~
\begin{itemize}
\item $\mathcal{DL}$ is not closed under MacNeille completions
\cite{Funayama44}.
\item $\mathcal{HA}$ and $\mathcal{BA}$ are 
closed under MacNeille completions.
\item $\mathcal{BA}$ is 
the only nontrivial proper subvariety of $\mathcal{HA}$ closed under MacNeille
completions
\cite{HB04}.
\end{itemize}
\end{thm}

As is well known, completion is a standard algebraic way 
to prove conservativity of extending first order logics to 
higher order ones. The above result indicates
that MacNeille completions 
work for classical and intuitionistic logics, but not 
for proper intermediate logics. On the other hand, one finds 
many varieties closed under MacNeille completions
when one moves to the realm of \emph{substructural logics} \cite{CGT12}.
See \cite{TV07} for a comprehensive account on the MacNeille completions.

Now an easy but crucial observation follows.

\begin{prop} \label{p-macneille}
A completion $\m A \subseteq \m B$ is MacNeille 
iff the rules below are valid:
$$
\infer{x\leq y}{ \{ a\leq y\}_{a \leq x}} \qquad
\infer{x\leq y}{ \{ x\leq a\}_{y \leq a}} 
$$
where $x,y$ range over $B$ and $a$ over $A$.
\end{prop}

The left rule has infinitely many premises indexed by 
the set $\{a \in A: a \leq x\}$. It states that 
if $a \leq x$ implies $a\leq y$ for every $a\in A$,
then we may conclude $x\leq y$.
This is valid just in case $x = \bigvee \{a \in A: a \leq x\}$.
Likewise, the right rule states that 
if $y\leq a$ implies $x\leq a$ for every  $a\in A$,
then $x\leq y$.
This is valid just in case $y = \bigwedge \{a \in A: y \leq a\}$.

As we will see,
the above looks very similar to the $\Omega$-rule.
This provides a link between lattice theory and proof theory.

\section{Takeuti correspondence between logic and arithmetic}
\label{s-takeuti}

There is a tight correspondence between systems of higher order logics
and those of arithmetic. A well-known example in computer science is 
that the numerical functions representable in System F 
coincide with the provably total functions of the second order 
Peano arithmetic $\ZT$. In this paper, we rather focus on another type
of correspondence, which we call the Takeuti correspondence, that
concerns with cut elimination in logic and 1-consistency in arithmetic.
One of our goals is to provide an easily accessible proof to 
the correspondence between
the arithmetical theories of iterated inductive definitions 
(up to $\omega$) and the parameter-free fragments of $\LIT$.

We first give some background on arithmetic and second order
logics (in Subsections \ref{ss-arithmetic} and \ref{ss-logic}),
then introduce the parameter-free systems and examine their expressive power
(in Subsections \ref{ss-parameter} and \ref{ss-higher}).

\subsection{Arithmetic}\label{ss-arithmetic}
Let $\IS$, $\PA$ and $\ZT$ be respectively the first order arithmetic 
with $\Sigma^0_1$ induction,
that with full induction,
and the second order 
arithmetic with full induction and comprehension (also called $\mathbf{Z}_2$).
Given a theory $T$ of arithmetic,
$T[X]$ denotes the extension of $T$ 
with a fresh set variable $X$ and atomic formulas
of the form $X(t)$. An expression of the form $\lambda x.\varphi(x)$ with
$\varphi$ a formula and $x$ a variable 
is called an  \emph{abstract}. Given an abstract
$\tau = \lambda x.\varphi(x)$ and a term $t$,
$\tau(t)$ stands for the formula $\varphi(t)$.

A great many subsystems of 
$\ZT$ are considered in the literature.
For instance, the system $\PICA$ is obtained by
restricting the induction and comprehension axiom schemata to 
$\Pi^1_1$ formulas. 
Even weaker are the
theories of iterated inductive definitions $\ID_n$
with $n <\omega$, that are obtained as follows.

$\ID_0$ is just $\PA$.
To obtain $\ID_{n+1}$, consider a formula 
$\varphi(X, x)$ in $\ID_n[X]$ 
which contains no free variables other than $X$ and $x$,
and no negative occurrences of $X$.
It defines a monotone map 
$\varphi^\Nat : \wp(\Nat) \longrightarrow \wp(\Nat)$ 
sending a set $X\subseteq \Nat$ to 
$\{n \in \Nat : \mathbb{N}\models \varphi(X, n)\}$.
By the Knaster-Tarski fixed point theorem, $\varphi^\Nat$ has 
a least fixed point $I^{\mathbb{N}}_\varphi$.
Hence it is reasonable to add a unary predicate symbol
$I_{\varphi}$ for each such $\varphi$ to the language of $\ID_n$ and axioms 
$$
\begin{array}{ll}
(\Lfp_1) & \varphi(I_{\varphi}) \subseteq I_{\varphi},\\
(\Lfp_2) & \varphi(\tau) \subseteq \tau \rightarrow I_{\varphi} \subseteq \tau,\\
\end{array}
$$
for every abstract $\tau = \lambda x.\xi(x)$ in the new language.
Here $\varphi(I_{\varphi})$ is a shorthand for the abstract 
$\lambda x.\varphi(I_\varphi, x)$ and $\tau_1 \subseteq \tau_2$ is 
for $\forall x.\tau_1(x) \rightarrow \tau_2(x)$. The induction schema 
is extended to the new language.
This defines the system $\ID_{n+1}$. Notice that 
$\ID_{n+1}$ does not involve any set variable. It is purely a 
first order theory of arithmetic.
Finally, let $\ID_{<\omega}$ be the union of all $\ID_{n}$ with
$n <\omega$.

Clearly $\ID_{<\omega}$ can be seen as a subsystem of $\PICA$.
In fact, any fixed point atom
$I_\varphi(t)$ can be replaced by its second order definition
$$
\Fix'_{\varphi}(t) \Def
\forall X. \forall x(\varphi(X, x) \rightarrow X(x))
\rightarrow X(t).
$$
This makes the axioms of $\ID_{<\omega}$ all provable in $\PICA$.

The converse is not strictly true, but 
it is known that $\ID_{<\omega}$ 
has the same proof-theoretic strength 
and the same arithmetical consequences with $\PICA$ (see \cite{BFPS81}).

Let us point out that typical use of an inductive definition is 
to define a provability predicate. Let $T$ be a sequent calculus system,
and suppose that we are given a formula $\varphi(X,x)$ 
saying that there is a rule in $T$ with 
conclusion sequent $x$ (coded by a natural number) and premises 
$Y\subseteq X$. Then $I_\varphi^{\mathbb{N}}$ gives the set of all provable sequents 
in $T$. Notice that the premise set $Y$ can be infinite. It is for this reason
that $\ID$-theories are suitable metatheories for infinitary
proof systems. See \cite{BFPS81} for more on inductive definitions.

Finally, let 
$\HAT$ be the second order Heyting arithmetic, and 
$\ID_n^i$ the intuitionistic counterpart of 
$\ID_n$ obtained by changing the underlying logic to the 
first order intuitionistic logic. Thus $\ID_0^i = \HA$ 
(the first order Heyting arithmetic).
The following result is well known (see \cite{BFPS81} for the 
second statement).

\begin{thm}\label{t-ci}
$\ZT$ and $\HAT$ prove exactly 
the same $\Pi^0_2$ sentences. 
Hence the 1-consistency of $\ZT$ is equivalent to 
that of $\HAT$ (provably in $\IS$).
The same holds for 
$\ID_n$ and $\ID_n^i$ for every $n<\omega$.
\end{thm}

Notice that the statement of 1-consistency (any provable $\Sigma^0_1$
sentence is true) and that of 
cut elimination are both $\Pi^0_2$.
Hence it does not matter much for our purpose 
whether the logic is classical or intuitionistic.

\subsection{Second order intuitionistic logic}\label{ss-logic}

In this subsection, we formally introduce sequent calculus 
$\LIT$ for the second order intuitionistic logic 
with full comprehension, that is an intuitionistic counterpart 
of Takeuti's classical calculus $\GLC$ for the second order 
classical logic
\cite{Takeuti53}.

Consider a language $L$ that consists of 
(first order) function symbols and predicate symbols.
A typical example is the language $L_{\PA}$ of Peano arithmetic,
which contains a predicate symbol $=$, 
constant $0$, successor $s$
and function symbols for all primitive recursive 
functions. Let
\begin{itemize}
\item $\Var$: a countable set of \emph{term variables} $x, y, z, \dots$,
\item $\VAR$: a countable set of \emph{set variables}
$X, Y, Z, \dots$,
\item $\Tm (L)$: the set of first order terms 
$t, u, v, \dots$ over $L$.
\end{itemize}

The set $\FM (L)$ of second order formulas is defined by:
$$
\varphi, \psi \bnf p(\vec{t}) \sep X(t) \sep \bot 
\sep \varphi \star \psi 
\sep Q x. \varphi 
\sep Q X. \varphi,
$$
where $\vec{t}$ is a list $t_1, \dots, t_n$ of terms over $L$,
$p$ is an $n$-ary predicate symbol in $L$, 
$\star \in \{\wedge, \vee, \rightarrow\}$ and $Q \in \{\forall, \exists\}$.
We define $\top := \bot\rightarrow\bot$.
When the language $L$ is irrelevant, we write 
$\Tm := \Tm(L)$ and 
$\FM := \FM(L)$.
Given $\varphi$, let 
$\FV(\varphi)$ and $\Fv(\varphi)$ be the set of 
free set variables and that of free term variables in $\varphi$, respectively.


We assume the standard variable convention that 
$\alpha$-equivalent formulas are syntactically identical, 
so that substitutions can be 
applied without variable clash.
A \emph{term substitution} is a function $\circ:
\Var\longrightarrow \Tm$.
Given $\varphi \in \FM$, the substitution instance $\varphi^\circ$
is defined as usual.
Likewise, 
a \emph{set substitution} is a function $\bullet: 
\VAR\longrightarrow \ABS$, where 
$\ABS := \{\lambda x.\varphi : \varphi \in \FM\}$.
Instance $\varphi^\bullet$ is obtained by 
replacing each atomic formula
$X(t)$ with $X^\bullet(t)$. 

A \emph{sequent of $\LIT$} is of the form $\Gamma\Rightarrow \Pi$,
where $\Gamma$ is a finite set of $\LIT$-formulas and 
and $\Pi$ is either the empty set or a singleton 
of an $\LIT$-formula.
We write $\Gamma, \Delta$ to denote $\Gamma\cup\Delta$.
The inference rules of $\LIT$ are given in Figure \ref{f-ruleLIT}.
We write $\LIT \vdash \Gamma \Rightarrow \Pi$ (resp.\ 
$\LIT \vdash^{cf} \Gamma \Rightarrow \Pi$) if 
the sequent $\Gamma\Rightarrow\Pi$ is provable 
(resp.\ cut-free provable) in $\LIT$.

\begin{figure}
\begin{screen}
$$
\begin{array}{ll}
\infer[(\mathsf{id})]{\Gamma,\varphi\Rightarrow\varphi}{} &
\infer[(\Cut)]{\Gamma\Rightarrow\Pi}{\Gamma\Rightarrow \varphi
& \varphi, \Gamma\Rightarrow\Pi}  \\[1em]
\infer[(\bot \Left)]{\bot, \Gamma \Rightarrow \Pi}{}
&
\infer[(\bot \Right)]{\Gamma \Rightarrow \bot}{\Gamma \Rightarrow}\\[1em]
\infer[(\wedge \Left)]{
\varphi_1 \wedge \varphi_2, \Gamma\Rightarrow\Pi}{
\varphi_i, \Gamma\Rightarrow\Pi}
&
\infer[(\wedge \Right)]{\Gamma\Rightarrow\varphi_1 \wedge\varphi_2}{
\Gamma\Rightarrow \varphi_1 & \Gamma\Rightarrow \varphi_2}\\[1em]
\infer[(\vee \Left)]{
\varphi_1 \vee \varphi_2, \Gamma\Rightarrow\Pi}{
\varphi_1,\Gamma\Rightarrow\Pi & 
\varphi_2,\Gamma\Rightarrow\Pi} 
&
\infer[(\vee \Right)]{\Gamma\Rightarrow\varphi_1 \vee\varphi_2}{
\Gamma\Rightarrow \varphi_i} \\[1em]
\infer[(\rightarrow \Left)]{
\varphi_1 \rightarrow \varphi_2, \Gamma\Rightarrow\Pi}{
\Gamma\Rightarrow\varphi_1 & 
\varphi_2,\Gamma\Rightarrow\Pi} 
&
\infer[(\rightarrow \Right)]{\Gamma\Rightarrow
\varphi_1 \rightarrow \varphi_2}{
\varphi_1,\Gamma\Rightarrow \varphi_2} \\[1em]
\infer[(\forall x\Left)]{
\forall x.\varphi(x), \Gamma\Rightarrow \Pi}{
\varphi(t),\Gamma\Rightarrow\Pi}
&
\infer[(\forall x\Right)]{
\Gamma\Rightarrow \forall x.\varphi(x)}{
\Gamma\Rightarrow \varphi(y) & y \not\in \Fv(\Gamma)} \\[1em]
\infer[(\exists x\Left)]{
\exists x.\varphi(x), \Gamma\Rightarrow \Pi}{
\varphi(y),\Gamma \Rightarrow\Pi & y\not\in \Fv(\Gamma,\Pi)}
& 
\infer[(\exists x\Right)]{
\Gamma\Rightarrow \exists x.\varphi(x)}{
\Gamma\Rightarrow \varphi(t)} \\[1em]
\infer[(\forall X\Left)]{
\forall X.\varphi(X), \Gamma\Rightarrow \Pi}{
\varphi(\tau),\Gamma\Rightarrow\Pi}
&
\infer[(\forall X\Right)]{
\Gamma\Rightarrow \forall X.\varphi(X)}{
\Gamma\Rightarrow \varphi(Y) & Y \not\in \FV(\Gamma)} \\[1em]
\infer[(\exists X\Left)]{
\exists X.\varphi(X), \Gamma\Rightarrow \Pi}{
\varphi(Y),\Gamma \Rightarrow\Pi & Y\not\in \FV(\Gamma,\Pi)}
& 
\infer[(\exists X\Right)]{
\Gamma\Rightarrow \exists X.\varphi(X)}{
\Gamma\Rightarrow \varphi(\tau)} 
\end{array}
$$
\end{screen}
\caption{Inference rules of $\LIT$} \label{f-ruleLIT}
\end{figure}


In the sequel, we will build parameter-free logical systems 
upon the first order fragment of $\LIT$.
Let $\Fm \subseteq \FM$ be the set of formulas without 
second order quantifiers. The ordinary sequent calculus 
$\LI$ for the first order intuitionistic logic can be obtained from $\LIT$
by restricting the formulas to $\Fm$ and by removing the rules 
for the second order quantifiers.

It is well-known that the cut elimination theorem for $\GLC$ or $\LIT$
implies 
the consistency of $\ZT$ or $\HAT$ finitistically 
(or in $\IS$, formally speaking).
We also have the converse, when 
consistency is replaced with \emph{1-consistency} meaning that
all provable $\Sigma^0_1$ sentences are true
(also called \emph{$\Sigma^0_1$-soundness}).


\begin{thm}\label{t-takeuti}
Let $\mathsf{CE}(\GLC)$ be a $\Pi^0_2$ sentence stating that 
$\GLC$ admits cut elimination, and
$\mathsf{1CON}(\ZT)$ a $\Pi^0_2$ sentence stating that $\ZT$ is 1-consistent.
Then:
$$
\IS \vdash \mathsf{CE}(\GLC) \leftrightarrow 
\mathsf{1CON}(\ZT).$$
\end{thm}

Actually the above theorem holds even if cut elimination is 
replaced with \emph{partial} cut elimination saying that 
any sequent $\Gamma\Rightarrow\Pi$ 
provable in $\LIP_n$ has a cut-free derivation 
\emph{provided that} $\Gamma\cup\Pi\subseteq \Fm$.

An even stronger result holds too,
as pointed out by \cite{Arai18} on the basis of P\"{a}ppinghous'
theorem \cite{Pappinghaus83}: \emph{complete} cut elimination 
is equivalent to 
 \emph{partial} cut elimination in the above sense.
Let  $\mathsf{CE}_{\Fm}(\GLC)$ be
a $\Pi^0_2$ sentence that expresses the statement of 
partial cut elimination for $\GLC$.

\begin{thm}\label{t-partial}
$\IS \vdash \mathsf{CE}(\GLC) \leftrightarrow \mathsf{CE}_{\Fm}(\GLC)$.
\end{thm}



\begin{rem}
The forward implication of Theorem \ref{t-takeuti}
is due to Takeuti \cite{Takeuti58b}, while the backward one 
is a folklore (see \cite{Girard87} and \cite{Arai18}). 

The same holds when $\ZT$ is replaced by $\HAT$ 
(because of Theorem \ref{t-ci}), and/or
$\GLC$ is replaced by $\LIT$
(because they admit essentially the same proof of cut elimination).

The paper \cite{Arai18} also mentions
the following correspondence:
$$
\IS \vdash \mathsf{CE}(\GLC(\Pi^1_n)) \leftrightarrow 
\mathsf{1CON}(\Pi^1_n\mbox{-{\bf CA}}_0)$$
for every $n<\omega$, where 
$\GLC(\Pi^1_n)$ is the fragment of $\GLC$ obtained by 
restricting the abstract $\tau$ in rules $(\forall X\Left)$ and
$(\exists X\Right)$ to $\Pi^1_n$ abstracts.
\end{rem}

This sort of correspondence 
between 1-consistency in arithmetic and cut elimination in logic
may be called the \emph{Takeuti correspondence}.
A goal of this paper is to provide Takeuti correspondences between
the theories $\HA, \ID_1^i, \ID_2^i, \dots$ of inductive definitions
and certain parameter-free fragments of $\LIT$ that are to be introduced next.



\subsection{Parameter-free second order intuitionistic logics}
\label{ss-parameter}

We here introduce fragments $\LIP_0, \LIP_1, \LIP_2, \dots$ of
$\LIT$. They are parameter-free because 
any formula of the form $\forall X. \xi$ or $\exists X.\xi$ 
is \emph{second order closed},
meaning that it does not contain any set parameter.

First, we write $\FMP_{-1}$ for $\Fm$ (the first order formulas)
for convenience. For each $n \geq 0$,
the set $\FMP_n$ of 
\emph{parameter-free formulas at level $n$} is defined by:
$$
\varphi, \psi \bnf
p(\vec{t}) \sep X(t) \sep \bot 
\sep \varphi \star \psi 
\sep Q x. \varphi
\sep Q X. \xi,
$$
where $\star \in \{\wedge, \vee, \rightarrow\}$, 
$Q \in \{\forall, \exists\}$ and 
$\xi$ is any formula in $\FMP_{n-1}$ such that 
$\FV(\xi) \subseteq \{X\}$. 
We let 
$$
\ABS_n := \{ \lambda x.\varphi : \varphi \in \FMP_n\}.
$$

An important property is the closure under substitution:
\begin{lem}
If $\varphi(X) \in \FMP_n$ and $\tau \in \ABS_n$,
then $\varphi(\tau) \in \FMP_n$.
\end{lem}

\begin{proof}
By induction on the structure of $\varphi(X)$. If it is an atom
$X(t)$, then $\varphi(\tau) = \tau(t) \in \FMP_n$.
The induction steps for first order connectives are easy.
If $\varphi = QY. \xi$, then it does not have a free occurrence of $X$,
so $\varphi(\tau) = \varphi \in \FMP_n$.
\end{proof}

Sequent calculus $\LIP_n$ for 
the \emph{parameter-free second order intuitionistic logic at 
level $n$}
is obtained from $\LIT$ by restricting the formulas 
to $\FMP_n$.
Most importantly, when one applies 
rules $(\forall X\Left)$ and $(\exists X\Right)$,
both the main formula $QX.\varphi$ and 
the minor formula $\varphi(\tau)$ must belong to $\FMP_n$.

Let $\FMP_{<\omega}$ be the union of all $\FMP_n$ and 
$\LIP_{<\omega}$ the sequent calculus associated to it.

\begin{rem}
The idea of restricting to the parameter-free formulas dates back to
\cite{Takeuti58}, which introduces a similar condition called 
isolatedness.
It also appears in more recent papers, such as \cite{Buchholz01,Altenkirch_Coquand01,Aehlig05}.
\end{rem}

A typical formula in $\FMP_0$ is
$$
\begin{array}{lll}
\Nn(t)  & :=  & \forall X.\,
Sub(X) \wedge Suc(X) \wedge X(0) \rightarrow X(t),
\end{array}
$$
where 
$Sub(X):= \forall xy.\, x=y \wedge X(x)\rightarrow X(y)$ and
$Suc(X) := \forall x.\, X(x) \rightarrow X(s(x))$.
Given a formula $\varphi$, let $\varphi^{\Nn}$ be
the formula obtained by replacing all first order quantifiers
$Q x$ with $Q x \in \Nn$. That is, we replace
$\forall x.\varphi$ with $\forall x.\, \Nn(x)\rightarrow\varphi$, and
$\exists x.\varphi$ with $\exists x.\, \Nn(x)\wedge\varphi$.
It is clear that if $\varphi$ is a first order formula,
then $\varphi^{\Nn}$ belongs to $\FMP_0$.

On the other hand, the standard
second order definitions of positive connectives $\{\exists,\vee\}$:
$$
\begin{array}{lll}
\exists X.\varphi(X) & := & \forall Y. \forall X (\varphi(X) \rightarrow Y(*))
\rightarrow Y(*),\\
\varphi\vee\psi & := & \forall Y. (\varphi \rightarrow Y(*))
\wedge (\psi \rightarrow Y(*)) \rightarrow Y(*),
\end{array}
$$
with $Y \not\in \FV(\varphi, \psi)$ and $*$ a constant, are no longer
available. They are not parameter-free
(unless $\varphi$ and $\psi$ are free of set variables in
the latter formula). Hence
restricting to the negative fragment 
$\{\forall, \wedge, \rightarrow\}$ causes a serious loss
of expressivity in our parameter-free setting.

\subsection{Expressivity of parameter-free logics}
\label{ss-expressivity}

Let us now briefly examine the expressivity of $\LIP_0$.
In the following, we consider 
terms and formulas over the language $L_{\PA}$.

It is not hard to see that $\LIP_0$ proves 
$$
\Nn(0), \qquad Suc(\Nn), \qquad
Sub(\Nn),\qquad \Gamma_{eq}\Rightarrow Sub(\tau),
$$
where 
$\tau$ is any abstract of the form
$\lambda x.\varphi^{\Nn}(x)$ with $\varphi \in \Fm$ 
and $\Gamma_{eq}$ consists of some equality axioms for predicate
and function symbols.
Moreover, the principle of 
mathematical induction is also available in $\LIP_0$.

\begin{lem}\label{l-ind}
$\LIP_0$ proves
$$
\Gamma_{eq}\Rightarrow [\forall x (\varphi(x) \rightarrow \varphi(s(x))) \wedge 
\varphi(0) \rightarrow  \forall y.\varphi(y)]^{\Nn}
$$
for every formula $\varphi$ in $\Fm$, where $\Gamma_{eq}$ is a
set of some equality axioms.
\end{lem}

\begin{proof}
Let $\tau := \lambda x.\varphi^{\Nn} (x) \wedge \Nn(x)$.
We claim that $\LIP_0$ proves
$$
(!)\qquad 
\Gamma_{eq},\,
[\forall x (\varphi(x)\rightarrow\varphi(s(x)))]^{\Nn},\,
\varphi^{\Nn}(0),\,
Sub(\tau) \wedge Suc(\tau) \wedge \tau(0)
\rightarrow \tau(y)
 \Rightarrow \varphi^{\Nn}(y).
$$
First, $\Gamma_{eq}\Rightarrow Sub(\tau)$ follows from 
$Sub(\Nn)$ and $\Gamma_{eq}\Rightarrow 
Sub(\lambda x.\varphi^{\Nn}(x))$, which are both 
provable.
Moreover, we can easily prove
$$
\begin{array}{l}
\varphi^{\Nn}(0)\Rightarrow \varphi^{\Nn}(0)\wedge \Nn(0),\\
{[\forall x.\, \varphi(x)\rightarrow\varphi(s(x))]}^{\Nn}
\Rightarrow 
\forall x.\, \varphi^{\Nn}(x) \wedge \Nn(x) \rightarrow 
\varphi^{\Nn}(s(x))\wedge \Nn(s(x)),
\end{array}
$$
by using $\Nn(0)$ and $Suc(\Nn) = 
\forall x.\, \Nn(x)\rightarrow\Nn(s(x))$.
Hence we have (!). Now the desired formula 
is obtained by $(\forall X\Left)$ and some elementary reasoning.
\end{proof}

Thus $\LIP_0$ can simulate reasoning in the first order Heyting arithmetic
$\HA$ (see Appendix \ref{a-proof1} for the detail).

\begin{thm}\label{t-cc}
$
\IS \vdash \mathsf{CE}(\LIP_0) \rightarrow 
\mathsf{1CON}(\HA)$.
\end{thm}

\begin{proof}
Suppose that $\HA$ proves 
a $\Sigma^0_1$ sentence $\varphi$.
We then have $\LIP_0 \vdash
\Gamma \Rightarrow \varphi$, where 
$\Gamma$ consists of some $\Pi^0_1$ axioms of $\PA$
(see Appendix \ref{a-proof1}).
Notice that the sequent only consists of first order formulas.
Hence assuming the (partial) cut elimination for $\LIP_0$, we obtain 
a cut-free derivation of it in $\LI$. By the standard soundness argument
one can verify that $\varphi$ is true. Moreover, all the reasoning
can be done in a finitistic way, so is formalizable in $\IS$.
\end{proof}


\subsection{Expressivity at higher levels}
\label{ss-higher}

We next proceed to the expressivity of $\LIP_n$ with $n>0$.
Consider the second order definition of
a least fixed point:
$$
\Fix_{\varphi}(t) \Def
\forall X.\, Sub(X)\wedge \forall x(\varphi(X, x) \rightarrow X(x))
\rightarrow X(t).
$$
This is parameter-free and belongs to $\FMP_n$ provided that 
$\varphi \in \FMP_{n-1}$ and $\FV(\varphi) \subseteq \{X\}$.
Moreover, it satisfies the axioms $(\Lfp_1)$ and $(\Lfp_2)$.

\begin{lem}\label{l-fix}
Let $\varphi(X,x)$ be a formula in $\FMP_{n-1}$ such that
$\FV(\varphi)\subseteq \{X\}$ and $X$ occurs only positively in it.
Then $\LIP_{n}$ proves
$$
\begin{array}{ll}
(\Lfp_1^{\Fix}) & \forall x.\varphi(\Fix_{\varphi}, x)\rightarrow \Fix_{\varphi}(x),\\
(\Lfp_2^{\Fix}) & 
\forall x.(\varphi(\tau, x)\rightarrow\tau(x) )
\rightarrow \forall y(\Fix_{\varphi}(y)\rightarrow\tau(y))
\end{array}
$$
for every $\tau \in \ABS_n$.
\end{lem}

\begin{proof}
For $(\Lfp_1^{\Fix})$, notice that $\LIP_{n}$ proves
$$Sub(X),\, \forall x. \varphi(X, x)\rightarrow X(x),\, \Fix_{\varphi}(y)
\Rightarrow X(y).$$
Since $X$ has only positive occurrences in $\varphi(X,x)$,
$$
Sub(X),\, \forall x. \varphi(X, x)\rightarrow X(x),\, \varphi(\Fix_{\varphi},y)
\Rightarrow \varphi(X,y)
$$
can be proved by induction on the structure of $\varphi$.
Hence
$$
Sub(X),\, \forall x. \varphi(X, x)\rightarrow X(x),\,
\varphi(\Fix_{\varphi},y) 
\Rightarrow X(y).
$$
From this, we obtain $(\Lfp_1^{\Fix})$ by rules
$(\forall x \Right)$,
$(\rightarrow\Right)$ and $(\forall X\Right)$.

$(\Lfp_2^{\Fix})$ is obtained from
$$
\LIP_n \vdash \forall x.(\varphi(\tau,x)\rightarrow \tau(x))
\rightarrow\tau(y),\, 
\forall x.(\varphi(\tau,x)\rightarrow \tau(x)) \Rightarrow \tau(y)
$$
by rule $(\forall X\Left)$ and some first order inferences.
\end{proof}

Based on this, we translate each $\ID_n$-formula $\varphi$ into
a formula $\varphi^{\Fix} \in \FMP_n$
such that $\FV(\varphi^{\Fix})=\FV(\varphi)$.
It proceeds by induction on $n$.
For $n=0$, we let $\varphi^{\Fix} := \varphi^{\Nn}$.
For $n>0$, we replace each fixed point atom
$I_\xi(t)$ of $\ID_n$ with $\Fix_{\xi^{\Fix}} (t)$,
where $\xi = \xi(X,x) \in \FMP_{n-1}$ and $\FV(\xi)\subseteq \{X\}$.
We also
replace each first order quantifier $Q x$ with $Q x \in \Nn$.


Theorem \ref{t-cc} can be generalized to an arbitrary level.

\begin{thm}\label{t-cc2}
For every $n<\omega$,
$
\IS \vdash \mathsf{CE}(\LIP_n) \rightarrow 
\mathsf{1CON}(\ID_n^i)$.
\end{thm}

\begin{proof}
$\LIP_n$ proves $Sub(\Fix_{\xi})$ for every $\xi\in \FMP_{n-1}$
with $\FV(\xi)\subseteq \{X\}$,
so proves $\Gamma_{eq}\Rightarrow Sub(\varphi^{\Fix})$ too for every $\ID_n$-formula $\varphi$.
Hence Lemma \ref{l-ind} can be extended to all formulas of the form
$\varphi^{\Fix} \in \FMP_{n}$. Furthermore, $\LIP_n$
proves 
$(\Lfp_1^{\Fix})$ and $(\Lfp_2^{\Fix})$ by Lemma \ref{l-fix}.
Thus $\LIP_n$ can simulate reasoning in $\ID_n^i$.
Therefore we can argue as in the proof of Theorem \ref{t-cc}.
\end{proof}

The converse implication can be obtained by
proving cut elimination for $\LIP_n$ ``locally'' within $\ID_n^i$,
that is, by proving
$$
\LIP_n \vdash \Gamma\Rightarrow\Pi 
\quad\mbox{implies}\quad
\ID_n^i \vdash \mbox{``$\LI \vdash^{cf} \Gamma\Rightarrow\Pi$.''}
$$
Thus the claim is that $\ID_n^i$ proves cut elimination for $\LIP_n$
``sequent-wise.''
More precisely, one has to show that
for each derivation $\pi$ of $\Gamma\Rightarrow\Pi$ in $\LIP_n$, 
there is a derivation $\pi'$ in $\ID_n^i$ of a $\Sigma^0_1$  sentence saying that
$\LI \vdash^{cf} \Gamma\Rightarrow\Pi$. 
Moreover, $\pi'$ should be obtained from $\pi$ primitive recursively.

Thus assuming that $\ID_n^i$ is 1-consistent,
we obtain a statement of cut elimination:
$$
\LIP_n \vdash \Gamma\Rightarrow\Pi 
\quad\mbox{implies}\quad
\LI \vdash^{cf} \Gamma\Rightarrow\Pi.
$$

This motivates us to prove cut elimination for parameter-free logics
\emph{locally} within $\ID$-theories.

As before, it is sufficient to prove \emph{partial} cut elimination
to establish the Takeuti correspondence. Moreover, 
Theorem \ref{t-partial} holds for $\LIP_n$ too, 
since 
the argument by P\"{a}ppinghaus \cite{Pappinghaus83} can be
restricted to a parameter-free fragment without any problem.

\begin{thm}\label{t-partial2}
For every $n<\omega$, 
$\IS \vdash \mathsf{CE}(\LIP_n) \leftrightarrow \mathsf{CE}_{\Fm}(\LIP_n)$.
\end{thm}

We are thus led to proving \emph{partial} cut elimination,
that is often simpler than proving \emph{complete} cut elimination.

\section{$\Omega$-rule}\label{s-omega}

\subsection{Introduction to $\Omega$-rule}
\label{ss-intro}

Cut elimination in a higher order setting is tricky, since a principal 
reduction step
$$
\infer[(\Cut)]{\Gamma \Rightarrow \Pi}{
\infer[(\forall X\Right)]{\Gamma \Rightarrow \forall X. \varphi(X)}{\Gamma \Rightarrow \varphi(Y)}
&
\infer[(\forall X\Left)]{\forall X.\varphi(X)\Rightarrow \Pi}{\varphi(\tau)
\Rightarrow \Pi}}
\qquad
\raisebox{1em}{$\Longrightarrow$}
\qquad
\infer[(\Cut)]{\Gamma \Rightarrow \Pi}{\Gamma \Rightarrow 
\varphi(\tau)
& 
\varphi(\tau)\Rightarrow \Pi}
$$
may yield a bigger cut formula so that one cannot simply argue by induction
on the complexity of the cut formula. The 
$\Omega$-rule, introduced by \cite{Buchholz81a},
is an alternative of rule 
$(\forall X\Left)$ that allows us to circumvent this difficulty.

For illustration, let us first consider a naive 
implementation of the $\Omega$-rule into our setting.
We extend the first order calculus $\LI$ 
by enlarging the set of formulas to $\FMP_0$ and by adding rules 
$(\forall X\Right)$ and 
$$
\infer[(\Omega_0 \Left)]{\forall X.\varphi, \Gamma\Rightarrow \Pi}{
\{\;\Delta,\Gamma\Rightarrow\Pi\;\}_{\Delta \in |\forall X.\varphi|_0}}
$$
where $|\forall X.\varphi|_0$ consists of finite sets
$\Delta \fsubseteq \Fm$ such that 
$\LI \vdash^{cf} \Delta \Rightarrow \varphi(Y)$ holds
for some $Y \not\in \FV(\Delta)$.

Rule $(\Omega_0 \Left)$ has infinitely many premises indexed by 
$|\forall X.\varphi|_0$. In this respect it looks similar to 
the characteristic rules of the MacNeille completions
(Proposition \ref{p-macneille}). In Section \ref{s-omegamacneille},
we will provide a further link between them.

$(\Omega_0 \Left)$ is as strong as rule
$(\forall X \Left)$ of $\LIP_0$.
To see this, consider a provable formula 
$\forall X.\varphi \Rightarrow \varphi (\tau)$ in $\LIP_0$.
Let $\Delta \in |\forall X.\varphi|_0$, that is, 
$\Delta \Rightarrow \varphi (Y)$ has a cut-free derivation $\pi_\Delta$
in $\LI$
for some $Y \not\in \FV(\Delta)$.
Then there is a derivation $\pi_\Delta^\tau$ 
of $\Delta \Rightarrow \varphi (\tau)$ in the extended system obtained 
by substituting $\tau$ for $Y$. Hence we have:
$$
\infer[(\Omega_0\Left)]{\forall X.\varphi \Rightarrow \varphi (\tau)}{
\{\;
\shortstack{$\qquad\vdots\;\pi_\Delta^\tau$ \\
$\Delta\Rightarrow\varphi(\tau)$}
\;\}_{\Delta \in |\forall X.\varphi|_0}}.
$$

Moreover, rule $(\Omega_0 \Left)$ suggests a natural step of 
cut elimination. Consider a cut:
$$
\infer[(\Cut)]{\Gamma \Rightarrow \Pi}{
\infer[(\forall X\Right)]{\Gamma \Rightarrow \forall X. \varphi(X)}{\Gamma \Rightarrow \varphi(Y)}
&
\infer[(\Omega_0 \Left)]{\forall X.\varphi\Rightarrow \Pi}{
\{\;
\shortstack{$\qquad\vdots\;\pi_\Delta$ \\
$\Delta\Rightarrow\Pi$}
\;\}_{\Delta \in |\forall X.\varphi|_0}}}
$$
Arguing inductively, assume that 
 $\Gamma \subseteq \Fm$ and $\Gamma \Rightarrow \varphi(Y)$ is 
cut-free provable in $\LI$. 
Then $\Gamma$ belongs to $|\forall X.\varphi|_0$,
so the conclusion $\Gamma\Rightarrow \Pi$ is just one of the 
infinitely many premises. Thus the above derivation reduces to:
$$
\infer*[\pi_\Gamma]{\Gamma\Rightarrow\Pi}{}.
$$

It looks fine so far. However, rule $(\Omega_0 \Left)$ cannot be combined with
the standard rules for the first order quantifiers.

\begin{prop}\label{p-counter}
System $\LI + (\forall X\Right) + 
(\Omega_0 \Left)$ is inconsistent.
\end{prop}

\begin{proof} 
Consider formula $\varphi := X(c)\rightarrow X(x)$  with 
$c$ a constant. We claim that $\forall X.\varphi \Rightarrow \bot$ 
is provable. Let $\Delta \in |\forall X.\varphi|_0$, that is, 
$\LI \vdash^{cf} \Delta \Rightarrow Y(c)\rightarrow Y(x)$ 
for some $Y \not\in \FV(\Delta)$. Notice that the sequent is first order,
$\Delta$ and $Y(c)\rightarrow Y(x)$ do not share any predicate 
symbol/variable,
and $Y(c)\rightarrow Y(x)$ is not provable.
Hence Craig's interpolation theorem
yields $\Delta \Rightarrow \bot$. From this,
$\forall X.\varphi \Rightarrow \bot$ follows by $(\Omega_0 \Left)$,
and so
$\exists x.\forall X.\varphi \Rightarrow \bot$.
On the other hand, 
$\Rightarrow\exists x.\forall X.\varphi$ is also provable.
Hence we obtain $\bot$.
\end{proof}

The primary reason for inconsistency is that 
$(\Omega_0 \Left)$ is not closed under term substitutions, while 
the standard treatment of first order quantifiers assumes that 
all rules are closed under 
term substitutions.
Hence we have to weaken first order quantifier rules 
to obtain a consistent system.
A reasonable way is to replace 
$(\forall x \Right)$ and 
$(\exists x\Left)$ with Sch\"{u}tte's
\emph{$\omega$-rules}, which are infinitary (see Figure \ref{f-ruleLIO}).

\begin{rem}
Buchholz' later paper  \cite{Buchholz01} includes a proof
of (partial) cut elimination for a parameter-free 
subsystem $\mathbf{BI}_1^-$ of analysis that can be understood
without recourse to ordinals.
It is extended to complete cut elimination for the same system by
\cite{Akiyoshi_Mints11}, and to complete cut elimination 
for $\PICA + \mathbf{BI}$ (bar induction) by \cite{Akiyoshi17}.
The $\Omega$-rule further finds applications in 
modal fixed point logics \cite{JS11,MS12}.
It is used to show the strong normalization for 
the parameter-free fragments of System F, provably in $\ID$-theories
\cite{Akiyoshi_Terui16}.

Our $(\Omega\Left)$ is a logical analog of Buchholz' rule. 
There is however a subtle difference.
The original rule has assumptions indexed 
by \emph{derivations} of $\Delta \Rightarrow \varphi(Y)$, not by
$\Delta$'s themselves. As an advantage, one obtains 
a concrete operator for cut 
elimination and reduces the complexity of inductive definition:
the original semiformal system can be defined by an inductive definition 
on a bounded formula, while ours requires a $\Pi^0_1$ formula.
However, this point is irrelevant to the subsequent argument.
\end{rem}

\subsection{Syntactic cut elimination by $\Omega$-rules}
\label{ss-syntactic}

We here give a syntactic proof of partial cut elimination for 
$\LIP_n$ for every $n\geq 0$. The crucial step is to
define an infinitary sequent calculus $\LIO_n$  for each $n$
based on the $\Omega$-rules.

Let $\LIO_{-1} := \LI$ for convenience. Provided that $\LIO_{n-1}$ has been
defined, the sequent calculus $\LIO_n$ is defined as follows.
Each sequent consists of formulas in $\FMP_n$, and the inference rules
are
$(\mathsf{id})$, $(\Cut)$, the rules for propositional connectives
in Figure \ref{f-ruleLIT} and the rules for quantifiers given 
in Figure \ref{f-ruleLIO}.

Some remarks are in order. First, notice that 
rules $(\forall x \Right)$ and 
$(\exists x\Left)$ are replaced with infinitary rules
$(\omega \Right)$ and 
$(\omega \Left)$. Second, $\LIO_n$ contains not just one, but all of 
$(\Omega_0\Left)$, \dots, $(\Omega_n \Left)$. 
Similarly for other $\Omega$-rules.
The reason is that $\LIO_{n}$ has to be an extension of 
$\LIO_{n-1}$. 
Notice that each index set $|\forall X.\varphi|_k$ consists of 
finite sets $\Delta \fsubseteq \FMP_{k-1}$, and 
$|\exists X.\varphi|_k$ consists of sequents
 $\Delta\Rightarrow\Lambda$ such that $\Delta\cup\Lambda \fsubseteq \FMP_{k-1}$.
Finally, $\LIO_n$ contains  superfluous 
rules $(\tilde{\Omega}_k \Left)$ and 
$(\tilde{\Omega}_k \Right)$ for each $k =0,\dots, n$, 
the former being derivable by combining 
$(\forall X\Right)$, $(\Omega_k \Left)$ and $(\Cut)$.
These are nevertheless included for a technical reason. 

\begin{figure}
\begin{screen}
$$
\begin{array}{ll}
\infer[(\forall x\Left)]{
\forall x.\varphi(x), \Gamma\Rightarrow \Pi}{
\varphi(t),\Gamma\Rightarrow\Pi}
&
\infer[(\omega \Right)]{
\Gamma\Rightarrow \forall x.\varphi(x)}{
\{\; \Gamma\Rightarrow \varphi(t)\;\}_{t\in\Tm}} \\[1em]
\infer[(\omega \Left)]{
\exists x.\varphi(x), \Gamma\Rightarrow \Pi}{
\{\;\varphi(t),\Gamma \Rightarrow\Pi\;\}_{t\in\Tm}}
& 
\infer[(\exists x\Right)]{
\Gamma\Rightarrow \exists x.\varphi(x)}{
\Gamma\Rightarrow \varphi(t)} \\[1em]
\infer[(\Omega_k\Left)]{\forall X.\varphi, \Gamma\Rightarrow\Pi}{
\{\;\Delta, \Gamma \Rightarrow \Pi\;\}_{
\Delta\in |\forall X.\varphi|_k}} 
&
\infer[(\forall X\Right)]{
\Gamma\Rightarrow \forall X.\varphi(X)}{
\Gamma\Rightarrow \varphi(Y) & Y \not\in \FV(\Gamma)}  \\[1em]
\infer[(\exists X\Left)]{
\exists X.\varphi(X), \Gamma\Rightarrow \Pi}{
\varphi(Y),\Gamma \Rightarrow\Pi & Y\not\in \FV(\Gamma,\Pi)}
& 
\infer[(\Omega_k \Right)]{\Gamma\Rightarrow\exists X.\varphi}{
\{\;\Gamma, \Delta \Rightarrow \Lambda\;\}_{
(\Delta\Rightarrow\Lambda) \in |\exists X. \varphi|_k}} 
\end{array}
$$
$$
\begin{array}{l}
\infer[(\tilde{\Omega}_k\Left)]{\Gamma\Rightarrow\Pi}{
\Gamma \Rightarrow \varphi(Y) &
\{\;\Delta, \Gamma \Rightarrow \Pi\;\}_{\Delta\in 
|\forall X.\varphi|_k} & Y \not\in \FV(\Gamma)} \\[1em]
\infer[(\tilde{\Omega}_k\Right)]{\Gamma\Rightarrow\Pi}{
\{\;\Gamma, \Delta \Rightarrow \Lambda\;\}_{
(\Delta\Rightarrow\Lambda) \in |\exists X. \varphi|_k}
 & \varphi(Y),\Gamma\Rightarrow\Pi & Y \not\in \FV(\Gamma,\Pi)} 
\end{array}
$$
where $k=0, \dots, n$ and 
$$
\begin{array}{lll}
|\forall X.\varphi (X)|_k & := & \{ \Delta : 
\LIO_{k-1} \vdash^{cf} \Delta \Rightarrow \varphi (Y)
\mbox{ for some } Y \not\in \FV(\Delta) \} \\
|\exists X.\varphi (X)|_k & := & \{ (\Delta\Rightarrow\Lambda) : 
\LIO_{k-1} \vdash^{cf} \varphi (Y), \Delta \Rightarrow \Lambda
\mbox{ for some } Y \not\in \FV(\Delta,\Lambda) \}.
\end{array}
$$
\end{screen}
\caption{Inference rules of $\LIO_n$ for quantifiers}\label{f-ruleLIO}
\end{figure}

The partial cut elimination theorem will be established by a series of 
lemmas.

\begin{lem}[Substitution]\label{l-subst}
Let $n \geq 0$.
$\LIO_n \vdash \Gamma \Rightarrow \Pi$ implies
$\LIO_n \vdash \Gamma^\bullet \Rightarrow \Pi^\bullet$
for every
set substitution $\bullet : \VAR \longrightarrow \ABS_n$.
\end{lem}

\begin{proof}
By induction on $n$ and on the structure of the derivation.
Let us treat only two cases.\\

\noindent
(1) The derivation ends with
$$
\infer[(\forall X\Right)]{\Gamma \Rightarrow
\forall X.\varphi(X)}{\Gamma \Rightarrow \varphi(Y)}.
$$
Update the given substitution $\bullet$ by letting
$Y^\bullet := Z$ (fresh variable), 
so that $Z \not\in \FV(\Gamma^\bullet)$. By the induction hypothesis
we have $\Gamma^\bullet\Rightarrow \varphi(Z)$, 
noting that $\FV(\varphi(Y))\subseteq \{Y\}$.
We therefore obtain 
$\Gamma^\bullet \Rightarrow \forall X. \varphi(X)$ as required.\\

\noindent
(2) The derivation ends with
$$
\infer[(\Omega_k \Left)]{\forall X.\varphi, \Gamma \Rightarrow
\Pi}{
\{\; \Delta, \Gamma \Rightarrow \Pi\;\}_{
\Delta \in |\forall X.\varphi|_k}}
$$
where $k\leq n$.
Let $\Delta \in |\forall X.\varphi|_k$,
that is, $\LIO_{k-1} \vdash \Delta \Rightarrow \varphi(Y)$ for some 
$Y \not\in \FV(\Delta)$.
We write
$\Delta = \Delta (X_1, \dots, X_m)$ indicating all free set variables 
occurring in $\Delta$. Let 
$\Sigma := \Delta(Z_1, \dots, Z_m)$, where 
variables $Z_1, \dots, Z_m$ are fresh.
We still have $\Sigma \in |\forall X.\varphi|_k$ by the induction hypothesis
on $n$.
Hence 
$\Sigma, \Gamma \Rightarrow \Pi$ is among the 
premises. Now update the substitution $\bullet$ 
by letting $Z_i^\bullet := X_i$ for $i=1, \dots, m$. We then have
$\Delta, \Gamma^\bullet \Rightarrow \Pi^\bullet$
by the induction hypothesis on the derivation. Since this holds for every
$\Delta \in |\forall X.\varphi|_k$, we obtain
$\forall X.\varphi, \Gamma^\bullet \Rightarrow \Pi^\bullet$
by rule $(\Omega_k \Left)$.
\end{proof}

\begin{lem}[Embedding]\label{l-emb}
$\LIP_n \vdash \Gamma\Rightarrow \Pi$ 
implies 
$\LIO_n \vdash \Gamma^\circ \Rightarrow \Pi^\circ$ 
for every term substitution $\circ : \Var \longrightarrow \Tm$.
\end{lem}

\begin{proof}
By structural induction on the derivation.
We only consider two cases.\\

\noindent
(1) The derivation ends with 
$$
\infer[(\forall x\Right)]{\Gamma\Rightarrow\forall x.\varphi(x)}{
\Gamma\Rightarrow \varphi(y) & y\not\in\Fv(\Gamma)}.
$$
By the induction hypothesis, we have 
$\Gamma^\circ \Rightarrow \varphi^\circ(t)$ for every $t\in \Tm$.
Hence $\Gamma^\circ \Rightarrow (\forall x.\varphi(x))^\circ$ is 
obtained by rule $(\omega \Right)$.\\

\noindent
(2) The derivation ends with $(\forall X\Left)$.
It suffices to show that 
$\LIO_n \vdash \forall X.\varphi(X)\Rightarrow \varphi(\tau)$
for any $\varphi\in \FMP_n$ and $\tau \in \ABS_n$. We are going to use rule $(\Omega_n\Left)$.
So let $\Delta \in |\forall X.\varphi|_n$, that is,
$\LIO_{n-1}\vdash \Delta \Rightarrow \varphi(Y)$ for some $Y\not\in \FV(\Delta)$.
Then $\LIO_{n}\vdash \Delta \Rightarrow \varphi(Y)$ since
$\LIO_n$ is an extension of $\LIO_{n-1}$, so 
$\LIO_{n} \vdash \Delta \Rightarrow \varphi(\tau)$ by Lemma \ref{l-subst}.
Hence we obtain the desired sequent by $(\Omega_n\Left)$.
\end{proof}


\begin{lem}[Cut elimination for $\LIO$]\label{l-cut}
$\LIO \vdash \Gamma\Rightarrow\Pi$ implies 
$\LIO \vdash^{cf} \Gamma\Rightarrow\Pi$.
\end{lem}

This can be proved by a rather standard means, because 
any principal cut between
$(\Omega\Left)$ and $(\forall X\Right)$, that is the most crucial case,
can be absorbed into 
rule $(\tilde{\Omega}\Left)$. 
A detailed proof will be given in Appendix \ref{a-proof2}.

\begin{lem}[Collapsing]\label{l-collapse}
$\LIO_n\vdash^{cf} \Gamma\Rightarrow \Pi$ implies 
$\LIO_{n-1}\vdash^{cf} \Gamma\Rightarrow \Pi$,
provided that $\Gamma\cup \Pi \subseteq \FMP_{n-1}$.
\end{lem}

\begin{proof}
By structural induction on the cut-free derivation of 
$\Gamma\Rightarrow\Pi$ in $\LIO_n$.\\

\noindent
(1) The derivation ends with 
$$
\infer[(\tilde{\Omega_n}\Left)]{\Gamma\Rightarrow\Pi}{
\Gamma \Rightarrow \varphi(Y) &
\{\;\Delta, \Gamma \Rightarrow \Pi\;\}_{\Delta\in 
|\forall X.\varphi|_n} & Y \not\in \FV(\Gamma)}.
$$
We have 
$\LIO_{n-1} \vdash^{cf} \Gamma \Rightarrow \varphi(Y)$
by the induction hypothesis, noting that 
$\Gamma\cup \{\varphi (Y)\} \subseteq \FMP_{n-1}$.
Hence $\Gamma \in |\forall X.\varphi|_n$,
so $\Gamma, \Gamma \Rightarrow 
\Pi$ is among the premises. Therefore
$\LIO_{n-1} \vdash^{cf} \Gamma \Rightarrow \Pi$ by the induction
hypothesis again. \\

\noindent
(2) The derivation ends with $(\tilde{\Omega_k}\Left)$
with $k<n$. It is straightforward from the induction hypotheses.\\

\noindent
(3) $n=0$ and the derivation ends with
$$
\infer[(\omega \Right)]{
\Gamma\Rightarrow \forall x.\varphi(x)}{
\{\; \Gamma\Rightarrow \varphi(t)\;\}_{t\in\Tm}}.
$$
We choose a variable $y$ such that $y \not \in \Fv(\Gamma)$.
We have $\LI\vdash^{cf} \Gamma \Rightarrow \varphi(y)$ by the 
induction hypothesis, so the conclusion sequent is 
obtained by 
$(\forall x \Right)$.

Other cases are treated similarly.
\end{proof}

Lemmas \ref{l-emb}, \ref{l-cut} and \ref{l-collapse} constitute a 
syntactic proof of partial cut elimination for $\LIP_n$.
From a metatheoretical point of view,
the most significant part is to define a provability predicate for
$\LIO_n$. For $n=-1$, a provability predicate
for $\LIO_{-1} = \LI$ can be defined in $\ID_0^i = \HA$ as usual, since 
the proof system is finitary.

For $n= 0$, observe that one can define a formula
$\mathsf{Step}(X, x)$ in $\HA[X]$ such that 
$$
\mathsf{Step}(X, \Code{\Gamma\Rightarrow\Pi}) 
\Iff \mbox{$\Gamma\Rightarrow\Pi$ is obtained from some $Y \subseteq X$ 
by applying a rule of $\LIO_0$,}
$$
where $\Code{~~}$ is a suitable coding function and 
$X$ is supposed to be a set of (the codes of) sequents.
Notice that the above formula relies on 
a provability predicate for $\LIO_{-1} = \LI$. 
Now let $\LLIO_0 := I_{\mathsf{Step}}$, that is available in $\ID_1^i$. We then have
$$
\LLIO_0 (\Code{\Gamma\Rightarrow\Pi})
\Iff \LIO \vdash \Gamma\Rightarrow\Pi.
$$

For $n>0$, a provability predicate $\LLIO_n$ can be defined
by relying on $\LLIO_{-1}, \dots, \LLIO_{n-1}$, thus in $\ID_{n+1}^i$.
Once suitable provability predicates have been defined, 
the rest of the proof can be smoothly formalized,
since it mostly proceeds by structural induction on the derivation
(see also Appendix \ref{a-proof2}).
Hence we obtain:

\begin{thm}[Syntactic cut elimination for $\LIP_n$]\label{t-ce}
Let $n\geq 0$ and $\Gamma\cup \Pi \subseteq \Fm$. Then
$\LIP_n \vdash \Gamma\Rightarrow\Pi$ implies 
$\LI\vdash^{cf} \Gamma\Rightarrow\Pi$. Moreover, this fact can be
proved in $\ID_{n+1}^i$.
\end{thm}

Observe that it is impossible to prove it
within $\ID_n^i$, because of 
Theorem \ref{t-cc2} and the second incompleteness theorem.

\section{$\Omega$-rule and MacNeille completion}
\label{s-omegamacneille}

In this section, we establish a formal connection between 
the $\Omega$-rule and the MacNeille completion. Let us start by 
introducing algebraic semantics for full second order calculus $\LIT$.

Let $L$ be a language.
A (complete) \emph{Heyting-valued prestructure} for $L$ is 
$\mc M = \langle \m A, M, \mc D, \mc L \rangle$ 
where 
$\m A = \langle A, \wedge, \vee, \rightarrow, 
\top, \bot\rangle$ is a complete Heyting algebra,
$M$ is a nonempty set (\emph{term domain}),
$\emptyset \neq \mc D \subseteq A^M$ (\emph{abstract domain}) and 
$\mc L$ consists of a function 
$f^{\mc M} : M^n \longrightarrow M$ for each $n$-ary 
function symbol $f \in L$
and 
$p^{\mc M} : M^n \longrightarrow A$
for each $n$-ary 
predicate symbol $p \in L$.
Thus $p^{\mc M}$ is an $\m A$-valued subset of $M^n$.

It is not our purpose to systematically develop a model theory for 
the intuitionistic logic. We will use prestructures only for
proving conservative extension and cut elimination. 
Hence we only consider \emph{term models} below, in which
$M= \Tm$ and $f^{\mc M} (\Vec{t}) = f(\Vec{t})$.
This assumption simplifies the interpretation of formulas a lot.

A \emph{valuation} on $\mc M$ is a function $\mc V : \VAR \longrightarrow
\mc D$. The \emph{interpretation} of formulas
$\mc V : \FM \longrightarrow \m A$ is inductively defined as follows:
$$
\begin{array}{llllll}
\mc V (p(\vec{t})) & := & p^{\mc M}(\Vec{t}) 
&
\mc V (X(t)) & := & \mc{V}(X)(t) \\
\mc V(\bot) & := & \bot & 
\mc V (\varphi \star \psi) & := & \mc V(\varphi) \star \mc V(\psi) \\
\mc V (\forall x.\varphi(x)) & := & \bigwedge_{t \in \Tm} \mc V(\varphi(t)) &
\mc V (\exists x.\varphi(x)) & := & \bigvee_{t \in \Tm} \mc V(\varphi(t))\\
\mc V (\forall X.\varphi) & := & \bigwedge_{F \in \mc D} \mc V[F/X](\varphi) &
\mc V (\exists X.\varphi) & := & \bigvee_{F \in \mc D} \mc V[F/X](\varphi)
\end{array}
$$
where $\star \in \{\wedge, \vee, \rightarrow\}$ and 
$\mc V[F/X]$ is an update of $\mc V$ that maps $X$ to $F$.
$\mc M$ is called a Heyting-valued \emph{structure} if 
$\mc V (\tau) \in \mc D$ holds for every valuation $\mc V$ and 
every abstract $\tau \in \ABS$.
Clearly $\mc M$ is a Heyting-valued structure if 
$\mc D = {\m A}^\Tm$. Such a structure 
is called \emph{full}.

Given a sequent $\Gamma \Rightarrow \Pi$,
let 
$$
\begin{array}{lll}
\mc V(\Gamma) & := & \bigwedge \{\mc V(\varphi) : \varphi \in \Gamma\},\\
\mc V(\Pi) & := & \bigvee \{\mc V(\psi) : \psi \in \Pi\}.
\end{array}
$$

It is routine to verify:
\begin{lem}[Soundness] \label{l-sound}
If $\LIT \vdash \Gamma\Rightarrow\Pi$, then 
$\Gamma\Rightarrow\Pi$ is \emph{valid}, that is, 
$\mc V(\Gamma^\circ) \leq \mc V(\Pi^\circ)$ holds for every valuation $\mc V$ 
on every Heyting structure $\mc M$ and every term substitution $\circ$.
\end{lem}

To illustrate the use of algebraic semantics, 
let us have a look at a proof of an elementary
fact that $\LIT$ is a conservative extension of $\LI$.

Let $\m L$ be the Lindenbaum algebra for $\LI$, that is,
$\m L := \langle \Fm/\!\!\! \sim, \wedge, \vee, 
\rightarrow, \top, \bot \rangle$ where 
$\varphi \sim \psi$ iff $\LI \vdash \varphi \leftrightarrow \psi$.
The equivalence class of $\varphi$ with respect to $\sim$ 
is denoted by $[\varphi]$.
$\m L$ is a Heyting algebra in which 
$$
(*)\qquad [\forall x.\varphi(x)] = \bigwedge_{t \in \Tm} [\varphi(t)], 
\qquad
[\exists x.\varphi(x)] = \bigvee_{t \in \Tm} [\varphi(t)]
$$
hold. Given a sequent $\Gamma\Rightarrow\Pi$,
elements $[\Gamma]$ and $[\Pi]$ in $\m L$ are naturally defined.

Let $\m G$ be a \emph{regular} completion of $\m L$.
Then $\mc M(\m{G}) := \langle \m G, \Tm, {\m G}^\Tm, \mc L\rangle$
is a full Heyting structure,
where $\mc L$ consists of a $\m G$-valued predicate $p^{\mc M (\m G)}$
defined by $p^{\mc M (\m G)} (\vec{t}) := [p(\vec{t})]$ for each $p\in L$
(in addition to the interpretations of function symbols).
Define a valuation $\mc I$ by
$\mc I (X)(t) := [X(t)]$. 
We then have 
$\mc I(\varphi) = [\varphi]$ for every $\varphi \in \Fm$
by regularity
(be careful here: $(*)$ may fail in $\m G$
if it is not regular).

Now, suppose that $\LIT$ proves $\Gamma \Rightarrow\Pi$ 
with $\Gamma\cup\Pi \subseteq \Fm$.
Then we have $\mc I(\Gamma) \leq \mc I(\Pi)$ by Lemma \ref{l-sound}, so 
$[\Gamma] \leq [\Pi]$, that is, $\LI \vdash \Gamma \Rightarrow\Pi$.
This proves that $\LIT$ is a conservative extension of $\LI$.

Although this argument cannot be fully formalized in $\HAT$ because of 
G\"{o}del's second incompleteness, 
it does admit a local formalization in $\ZT$. In contrast, the above 
argument, when applied to $\LIP_0$, 
cannot be locally formalized in the arithmetical counterpart $\HA$. 
The reason is simply that
$\HA$ does not have second order quantifiers,
which are needed to write down the definitions of
$\mc V (\forall X.\varphi)$ and 
$\mc V (\exists X.\varphi)$. To circumvent this, we will make a 
crucial observation that 
$\mc V (\forall X.\varphi)$ and
$\mc V (\exists X.\varphi)$ admit alternative first order definitions 
if the completion is MacNeille. It is here that one finds 
a connection between the MacNeille completion and the 
$\Omega$-rule.

\begin{thm}\label{t-crucial}
Let $\m L$ be the Lindenbaum algebra for $\LI$ and
$\m L \subseteq \m G$ a regular completion.
$\mc M(\m G)$ and $\mc I$ are defined as above.
For every 
\emph{sentence} $\forall X.\varphi$ in $\FMP_0$,
the following are equivalent.
\begin{enumerate}
\item $\mc I(\forall X.\varphi) = \bigvee \{ a\in \m L : 
a \leq \mc I(\forall X.\varphi)\}$.
\item $\mc I(\forall X.\varphi) = \bigvee \{ [\Delta] \in \m L : 
 \Delta \in |\forall X.\varphi|_0 \}$.
\item The inference below is sound for every 
$y \in \m G$:
$$
\infer{\mc I(\forall X.\varphi) \leq y}{
\{\; \mc I(\Delta) \leq y\;\}_{\Delta \in |\forall X.\varphi|_0 }}
$$
\end{enumerate}
If $\m G$ is the MacNeille completion of $\m L$, all the above
hold.
\end{thm}

\begin{proof}
((1) $\Leftrightarrow$ (2)) Let $a=[\Delta]$. It is sufficient to 
show
$$a \leq \mc I(\forall X.\varphi) \Iff \Delta \in 
|\forall X.\varphi|_0.$$

Suppose that $a \leq \mc I(\forall X.\varphi)
= \bigwedge_{F \in \m{G}^{\Tm}} \mc I[F/X](\varphi)$.
We choose $Y \not\in \FV(\Delta)$
and define $F_Y \in \m G^{\Tm}$ by $F_Y(t) := [Y(t)]$ for every $t\in \Tm$. 
We then have 
$$[\Delta] \leq \mc I(\forall X.\varphi) \leq
\mc I[F_Y/X](\varphi(X)) = [\varphi(Y)],$$ that is,
$\LI\vdash \Delta\Rightarrow \varphi(Y)$. By the cut elimination for $\LI$, 
we obtain $\LI\vdash^{cf} \Delta\Rightarrow \varphi(Y)$.
Hence $\Delta \in |\forall X.\varphi|_0$.

Conversely, suppose that $\LI\vdash^{cf} \Delta
\Rightarrow \varphi(Y)$ with $Y\not\in \FV(\Delta)$. It implies 
$[\Delta] = \mc I (\Delta) = \mc I [F/Y] (\Delta) \leq \mc I[F/Y](\varphi(Y))$ 
for every $F \in {\m G}^\Tm$ by Lemma \ref{l-sound}. 
Hence $[\Delta] \leq \mc I(\forall X.\varphi(X))$.\\

\noindent
((2) $\Rightarrow$ (3)) Assume the premises of (3).
This means that we have $[\Delta] = \mc I(\Delta) \leq y$
for every $\Delta \in |\forall X.\varphi|_0$. Hence 
the conclusion $\mc I(\forall X.\varphi) \leq y$ follows by (2).\\

\noindent
((3) $\Rightarrow$ (2))  Let $y:= \bigvee \{ [\Delta] \in \m L : 
 \Delta \in |\forall X.\varphi|_0 \}$.
Then $\mc I(\Delta) = [\Delta] \leq y$ holds for every 
$\Delta \in |\forall X.\varphi|_0$, so
$\mc I(\forall X.\varphi) \leq y$ by rule (3). 

On the other hand, 
$\Delta \in |\forall X.\varphi|_0$ implies 
$[\Delta] \leq \mc I(\forall X.\varphi)$ 
as proved above. Hence we also have
$y\leq \mc I(\forall X.\varphi)$.\\

\noindent
Finally, suppose that $\m L \subseteq \m G$ is a MacNeille completion.
Then (1) holds by $\bigvee$-density. So (2) and (3) hold too.
\end{proof}

The equivalence in Theorem \ref{t-crucial} is quite suggestive, 
since (3) is an algebraic interpretation of rule $(\Omega_0\Left)$, 
while (1) is a characteristic of the MacNeille completion.
Equation (2) suggests a way of interpreting second order formulas 
without using second order quantifiers at the meta-level.
All these are true if the completion is MacNeille. 

\begin{rem}
Essentially the same as (2) has been already observed
by Altenkirch and Coquand \cite{Altenkirch_Coquand01}
in the context of lambda calculus (without making any explicit connection to
the $\Omega$-rule and the MacNeille completion).
Indeed, they consider a logic which roughly amounts to the 
negative fragment of our $\LIP_0$ and employ equation (2) to give 
a ``finitary'' proof of a (partial) normalization theorem for 
a parameter-free fragment of System F (see also 
\cite{Aehlig08,Akiyoshi_Terui16} for extensions).
However, their argument is technically based on a
downset completion, that is not MacNeille. As is well known, 
such a naive completion does not 
work well for the positive connectives $\{\exists, \vee\}$. 
In contrast, 
when $\m L \subseteq \m G$ is a MacNeille completion,
we also have
$$\mc I(\exists X.\varphi) = \bigwedge 
\{ [\Delta]\rightarrow[\Lambda] \in \m L : 
(\Delta\Rightarrow\Lambda) \in |\exists X.\varphi|_0 \}.$$

We thus claim that the insight by Altenkirch and Coquand
is further augmented and better understood 
if one employs the MacNeille completion instead of the downset completion
(or the filter completion).
\end{rem}

As a consequence of Theorem \ref{t-crucial}, it is possible to
give an algebraic proof to the conservativity of 
$\LIP_0$ over $\LI$,
that can be locally formalized in $\HA$. 

The argument proceeds as follows.
Let $\m L$ be the Lindenbaum algebra for $\LI$
and $\m G$ be the \emph{MacNeille} completion of $\m L$.
Then $\mc M(\m{G}) := \langle \m G, \Tm, {\m G}^\Tm, \mc L\rangle$
is a full Heyting structure.
Define a valuation $\mc I$ by
$\mc I (X)(t) := [X(t)]$ as before. To extend it inductively to 
the $\FMP_0$ formulas, we use the clauses
$$
\begin{array}{lll}
\mc I(\forall X.\varphi) & := & \bigvee \{ [\Delta] \in \m L : 
 \Delta \in |\forall X.\varphi|_0 \} \\
\mc I(\exists X.\varphi) & := & \bigwedge 
\{ [\Delta]\rightarrow[\Lambda] \in \m L : 
(\Delta\Rightarrow\Lambda) \in |\exists X.\varphi|_0 \}.
\end{array}
$$

Soundness holds with respect to this interpretation 
by Theorem \ref{t-crucial}.
Hence by the same argument as before, we may conclude that 
$\LIT$ is a conservative extension of $\LI$. 
We will not discuss formalization in $\HA$ here, as 
stronger results on cut elimination will be formalized in Section 
\ref{s-algebraic}.

It is interesting to see
that the second order $\forall$ is interpreted by the first order $\bigvee$, 
while the second order $\exists$ is 
by the first order $\bigwedge$. We call this style of 
interpretation
the \emph{$\Omega$-interpretation}, that is the algebraic side
of the $\Omega$-rule, and 
that will play a key role 
in the next section. We conclude our discussion by reporting 
a counterexample for the general soundness.

\begin{prop}\label{p-counter2}
There is a Heyting-valued structure in which 
$(\Omega_0 \Left)$ is not sound.
\end{prop}
 
\begin{proof}
Let $\m A$ be the three-element chain $\{0 < 0.5 < 1\}$ seen as a Heyting 
algebra. Here the implication $\rightarrow$ is defined by:
$$
\begin{array}{llll}
a \rightarrow b & := & \top & \mbox{if $a\leq b$,}\\
 & := & b & \mbox{otherwose.}
\end{array}
$$
Consider the language that only consists of 
a term constant $*$.
Then a full Heyting-valued structure 
$\mc{A} := \langle \m A, \Tm, \m{A}^\Tm, \mc L\rangle$
is naturally obtained.
Let $\varphi := (X(*)\rightarrow\bot)\vee X(*)$.
We then have
$\mc V(\forall X.\varphi) = 0.5$ for every valuation $\mc V$.
In fact, $\mc V(\varphi) = 1$ 
if $\mc V(X(*))=0$ or $1$, and 
$\mc V(\varphi) = 0.5$ 
if $\mc V(X(*))=0.5$.

Now consider the following instance:
$$
\infer[(\Omega_0 \Left)]{\forall X.\varphi \Rightarrow \bot}{
\{\; \Delta\Rightarrow \bot\;\}_{\Delta \in |\forall X.\varphi|_0}}.
$$
We claim that it is not sound \emph{provided that}
$\mc V (X(t)) = 0$ for every $X \in \VAR$ and $t\in \Tm$.
Suppose that $\Delta \in |\forall X.\varphi|_0$, i.e., 
$\LI \vdash^{cf} \Delta \Rightarrow \varphi(Y)$ with $Y \not\in \FV(\Delta)$.
Then 
$\mc V(\Delta) \leq \mc V[F/X](\varphi)$ for every $F \in 
\m A^\Tm$ by Lemma \ref{l-sound}. Hence
$$\mc V(\Delta) \leq 
\mc \bigwedge_{F\in \m{A}^\Tm} \mc V[F/X](\varphi) =
\mc V(\forall X.\varphi) = 0.5.$$
But $\Delta$ is first order and does not involve any predicate symbol, 
so only takes value $0$ or $1$
by the assumption on $\mc V$ (and the fact that $\{0 < 1\}$ is 
a Heyting subalgebra of $\m A$).
Hence $\mc V(\Delta) = 0$, that is, 
all the premises $\Delta\Rightarrow \bot$
are satisfied. However, 
$\mc V(\forall X.\varphi) = 0.5 > 0$, that is, the conclusion 
$\forall X.\varphi \Rightarrow \bot$ is
not satisfied.
\end{proof}

This invokes a natural question. Is it possible to find 
a \emph{Boolean-valued} counterexample? In other words, is the $\Omega$-rule
classically sound? This question is left open.

\section{Algebraic cut elimination}\label{s-algebraic}

This section is devoted to an algebraic proof of cut elimination
for parameter-free logics. After introducing a general concept
of Heyting frame
in Subsection \ref{ss-frame}, we consider a syntactic frame
build upon cut-free provability in Subsection \ref{ss-algebraic}.
A soundness argument then establishes the cut elimination theorem 
in Subsection \ref{ss-sound}. A small improvement is given 
in Subsection \ref{ss-modification}, that will be important
when formalizing our proof in $\ID$-theories in Section 
\ref{s-formalize}. An algebraic proof of cut elimination
for $\LIT$ due to \cite{Maehara91,Okada02} is given in Appendix 
\ref{a-okada} for a comparison.

\subsection{Polarities and Heyting frames}\label{ss-frame}

We begin with a very old concept due to Birkhoff \cite{Birkhoff40}, 
that provides 
a uniform framework for both MacNeille completion and cut elimination.

A \emph{polarity} $\m W = \langle W, W', R\rangle$ (a.k.a.\ \emph{formal 
context}) consists of 
two sets $W, W'$ and a binary relation $R \subseteq W\times W'$.
Given $X \subseteq W$ and $Z \subseteq W'$, let
$$
X^\rhd \Def \{ z \in W':  x\mathrel{R} z \mbox{ for every } x\in X\},
\qquad
Z^\lhd \Def \{ x \in W:  x\mathrel{R} z
\mbox{ for every } z\in Z\}.
$$
For example, let $\mathbf{Q}:=\langle \mathbb{Q}, \mathbb{Q}, \leq\rangle$.
Then 
$X^\rhd$ is the set of upper bounds of $X$ and 
$Z^\lhd$ is the set of lower bounds of $Z$. Hence
$X^{\rhd\lhd}$ is the lower part of a Dedekind cut for every nonempty 
$X\subseteq \mathbb{Q}$ bounded above.

The pair $(\rhd, \lhd)$ forms a \emph{Galois connection}:
$$
X \subseteq Z^{\lhd} \quad\Longleftrightarrow\quad X^{\rhd} \supseteq Z
$$
so induces a closure operator  
$\gamma (X)  :=  X^{\rhd\lhd}$ on $\wp(W)$, that is,
$$X \subseteq \gamma(Y) \Iff \gamma(X)\subseteq\gamma(Y)$$ 
holds for any
$X, Y \subseteq W$. Note that $X\subseteq W$ is closed 
iff there is $Z \subseteq W'$ such that $X = Z^{\lhd}$.
In the sequel, we also make use of the property
$$
Z_1 \subseteq Z_2 \Implies Z_2^\lhd \subseteq Z_1^\lhd \qquad
(Z_1, Z_2 \subseteq W').
$$
We write $\gamma (x) := \gamma (\{x\})$,
$x^\rhd := \{x\}^{\rhd}$ and
$z^\lhd := \{z\}^{\lhd}$.
Let 
$$\Gal (\mathbf{W})  := \{X\subseteq W: X = \gamma(X)\},$$
$X \wedge Y := X\cap Y$, 
$ X \vee Y := \gamma(X\cup Y)$, 
$ \top := W$ and 
$\bot := \gamma(\emptyset)$.

\begin{lem}
If $\m W$ is a polarity, then 
$\mathbf{W}^+ := \langle \Gal(\mathbf{W}), \wedge, \vee \rangle$
is a complete lattice.
\end{lem}

This is just a well-known fact. See \cite{DP02} for instance.

The lattice $\mathbf{W}^+$ is not always distributive 
because of the use of $\gamma$ in the definition of $\vee$.
To ensure distributivity, we have to impose a further structure 
on $\m W$.


A \emph{Heyting frame} is
$\m W = \langle W, W', R, \circ, \varepsilon, \ldd \rangle$, where
\begin{itemize}
\item $\langle W, W', R \rangle$ is a polarity,
\item $\langle W, \circ, \varepsilon \rangle$ is a monoid,
\item function $\ldd : W \times W' \longrightarrow W'$ satisfies
$$
x \circ y \Rel z \Iff y \Rel x\ldd z
$$
for every $x,y \in W$ and $z\in W'$,
\item the following inferences are valid:
$$
\infer[(e)]{y \circ x \Rel z}{x\circ y \Rel z} 
\qquad
\infer[(w)]{x \Rel z}{\varepsilon \Rel z} 
\qquad
\infer[(c)]{x \Rel z}{x\circ x \Rel z} 
$$
\end{itemize}
Clearly $x \Rel z$ is an analogue of a sequent and 
$(e), (w)$ and $(c)$ correspond to exchange, weakening and contraction
rules. By removing some/all of them, one obtains a
\emph{residuated frame} that works for substructural logics as well
\cite{GJ13,CGT12}.

\begin{lem}\label{l-heyting}
If $\m W$ is a Heyting frame, 
$\mathbf{W}^+ := \langle \Gal(\mathbf{W}), \wedge, \vee, \rightarrow, 
\top, \bot \rangle$
is a complete Heyting algebra, where 
$X \rightarrow Y := \{ y \in W: x \circ y \in Y \mbox{ for every } x\in X\}$.
\end{lem}

\begin{proof}
First of all, observe that 
any $X\in \Gal(\m W)$ is closed under $(e)$, $(w)$ and $(c)$,
that is, the following inferences are all valid:
$$
\infer[(e)]{y \circ x \in X}{x \circ y \in X}
\qquad
\infer[(w)]{x \circ y \in X}{x \in X}
\qquad
\infer[(c)]{x \in X}{x\circ x\in X}
$$
We only verify $(w)$. Suppose that $x \in X$ and 
$z \in X^\rhd$. Then $x\Rel z$, i.e., $x\circ \varepsilon \Rel z$.
So $\varepsilon \Rel x \ldd z$ and 
$y \Rel x \ldd z$ by $(w)$. Hence $x \circ y \Rel z$.
Since this holds for every $z \in X^\rhd$, 
we conclude $x\circ y \in X^{\rhd\lhd} = X$.

Next, we show that $X\rightarrow Y \in \Gal(\m W)$ whenever
$Y \in \Gal(\m W)$. This can be shown by proving
$$
X\rightarrow Y = (X\ldd Y^\rhd )^{\lhd}
$$
where 
$X\ldd Y^\rhd := \{ x\ldd z \in W' : x\in X, z \in Y^\rhd\}$.

For the forward direction, let $y \in X\rightarrow Y$,
$x\in X$ and $z \in Y^\rhd$. Then $x \circ y \in Y$, 
so $x\circ y \Rel z$, hence $y \Rel x\ldd z$.
Since this holds for every $x\ldd z \in
X\ldd Y^\rhd$, we conclude $y \in (X\ldd Y^\rhd )^{\lhd}$.

For the backward direction, let $y \in 
(X\ldd Y^\rhd )^{\lhd}$, $x\in X$ and $z\in Y^\rhd$.
Then we have $y \Rel x\ldd z$, so $x\circ y \Rel z$.
Since this holds for every $z\in Y^\rhd$, we have 
$x\circ y \in Y^{\rhd\lhd} = Y$. Since this holds for every 
$x \in X$, we conclude $y \in X\rightarrow Y$.

We now prove that 
$$
X\cap Y \subseteq Z \Iff X \subseteq Y \rightarrow Z
$$
holds for every $X, Y, Z \in \Gal(\m W)$.
For the forward direction, let 
$x\in X$ and $y \in Y$. Then $x\circ y \in X\cap Y$ by $(e)$ and $(w)$,
so $x\circ y \in Z$ by assumption.
Since this holds for every $y\in Y$, we have 
$X \subseteq Y \rightarrow Z$.

For the backward direction, let $x\in X\cap Y$.
Then $x\circ x \in Z$ by assumption, so $x \in Z$ by $(c)$.
This proves $X\cap Y \subseteq Z$.
\end{proof}

Polarities and Heyting frames are handy devices to obtain MacNeille completions.
Let $\m A = \langle A, \wedge, \vee, \rightarrow, \top, \bot\rangle$ 
be a Heyting algebra. Then 
$\m{W}_{\m A} \Def \langle A, A, \leq, \wedge, \top, \rightarrow\rangle$
is a Heyting frame.
Notice that the third condition for a Heyting frame above amounts to
$x\wedge y \leq z$ iff $y \leq x\rightarrow z$.
For the next theorem,
note that the closure operator $\gamma$ can be seen as a map
$\gamma : \m A \longrightarrow \m{W}_{\m A}^+$ sending $a\in A$
to $a^{\rhd\lhd}$.

\begin{thm}\label{t-regular}
If $\m A$ is a Heyting algebra, then
$\gamma: \m A \longrightarrow 
\m{W}_{\m A}^+$ is a MacNeille completion.
\end{thm}

\begin{proof}
It is easy to see that $\gamma(a) = a^\lhd$ holds for every $a\in A$.
Based on this, we can show that $\gamma(\bot)=\gamma(\emptyset)$ and
$\gamma(a\star b) = \gamma(a) \star \gamma(b)$ for $\star \in \{\wedge,
\vee,\rightarrow\}$. Furthermore, $a\leq b$ holds if and only if
$\gamma(a)\subseteq\gamma(b)$, meaning that $\gamma$ is an embedding.

Let us verify that the completion is MacNeille.
Let $X \in \Gal(\m{W}_{\m A})$.
For $\bigvee$-density, we have
$$X = \gamma \left( \bigcup \{\gamma(a) : a \in X\} \right)
= \bigvee \{\gamma(a) : \gamma(a)\leq X\}.$$
For $\bigwedge$-density, notice that 
$X = \bigcap \{a^\lhd : a \in X^\rhd\}$ and
$\gamma (a) = a^\lhd$. Hence
$$
X = \bigwedge \{\gamma(a) : X \leq \gamma (a)\}.
$$
\end{proof}


\subsection{A syntactic frame for cut elimination}
\label{ss-algebraic}

We now start an algebraic proof of (partial) cut elimination for 
$\LIP_{n+1}$ (with $n\geq -1$). Although we have already given a proof of 
cut elimination in Subsection \ref{ss-syntactic}, the proof does not 
formalize in $\ID_{n+1}^i$ but only in $\ID_{n+2}^i$ (even locally).
Our goal here is to give another proof that 
locally formalizes in $\ID_{n+1}^i$.

What we actually do is to prove that
$$\LIP_{n+1} \vdash \Gamma\Rightarrow\Pi 
\quad\mbox{implies}\quad
\LIO_{n} \vdash^{cf} \Gamma\Rightarrow\Pi$$ 
provided that $\Gamma \cup \Pi \subseteq \FMP_n$.
In particular when $n= -1$, this means that 
$\LIP_{0} \vdash \Gamma\Rightarrow\Pi$ implies
$\LI \vdash^{cf} \Gamma\Rightarrow\Pi$ provided that 
$\Gamma \cup \Pi \subseteq \Fm$. 
When $n\geq 0$, we may combine it with 
Lemma \ref{l-collapse} 
to obtain partial cut elimination for $\LIP_{n+1}$. 
Notice that any use of a provability predicate at level $n+1$ is 
avoided here. It is for this reason that the argument locally
formalizes in $\ID_{n+1}^i$.

To begin with,
let $\fwp (\Fm)$ be the set of finite sets of first order formulas, 
so that $\langle \fwp (\Fm), \cup, \emptyset\rangle$ is a commutative 
idempotent monoid. 
Let $\SEQ$ be
the set of sequents that consist of formulas in $\FMP_n$. 
There is a natural 
map $\ldd : \fwp (\Fm) \times \SEQ \longrightarrow \SEQ$ 
defined by 
$\Gamma \ldd (\Sigma\Rightarrow \Pi) := (\Gamma,\Sigma \Rightarrow \Pi)$.
So
$$\m{CF} \Def \langle \fwp(\Fm), \SEQ, \Rightarrow^{cf}, \cup, \emptyset,
\ldd \rangle$$
is a Heyting frame, where the binary relation $\Rightarrow^{cf}$
is defined by 
$$\Gamma \Rightarrow^{cf} (\Sigma\Rightarrow\Pi)
\Iff
\LIO_n \vdash^{cf} \Gamma, \Sigma \Rightarrow \Pi.$$
In fact, rules $(e)$, $(c)$ are automatically satisfied because 
the monoid is commutative and idempotent. Rule $(w)$ is satisfied
since the weakening rule is admissible in $\LIO_n$.
Finally, we have:
$$
\Delta\cup\Gamma \Rcf (\Sigma\Rightarrow\Pi)
\Iff
\Delta \Rcf \Gamma \ldd(\Sigma\Rightarrow\Pi).
$$
In the following, we write $\varphi$ for 
sequent $(\emptyset \Rightarrow \varphi) \in \SEQ$.
Thus $\Gamma \Rightarrow^{cf}\varphi$ simply means 
$\LIO_n \vdash^{cf} \Gamma \Rightarrow \varphi$.

$\m{CF}$ is a frame in which 
$\Gamma \in \Pi^\lhd$ holds iff $\Gamma \Rightarrow^{cf} \Pi$.
In particular, $\varphi \in \varphi^\lhd$ always holds,
so 
$$(*)\qquad \varphi \in \gamma(\varphi) \subseteq \varphi^\lhd.$$
It should also be noted that each $X \in \Gal(\m{CF})$ is closed under 
weakening because of $(w)$: 
if $\Delta \in X$ and $\Delta \subseteq \Sigma$, then $\Sigma \in X$.

This yields a full Heyting-valued structure
$$\mc{CF} := \langle \m{CF}^+, \Tm, \Gal(\m{CF})^\Tm, \mc L
\rangle,$$
where $\mc L$ is defined by
 $p^{\mc{CF}} (\Vec{t}) := \gamma(p(\Vec{t}))$ for each predicate symbol
$p \in L$.

Let $\mc{I} : \VAR \longrightarrow \Gal(\m{CF})^\Tm$
be a valuation defined by 
$\mc{I}(X)(t) := \gamma (X(t))$. This can be extended to 
an interpretation
$\mc{I} : \FMP_0
\longrightarrow \Gal(\m{CF})$ by induction, employing 
the $\Omega$-interpretation technique discussed in 
Section \ref{s-omegamacneille}:
$$
\begin{array}{llllll}
\mc I(X(t)) & := & \mc{I} (X)(t) = \gamma (X(t)), &
\mc I(p(\vec{t})) & := & p^{\mc{CF}}(\vec{t}) = \gamma (p(\vec{t})), \\
\mc I(\bot) & := & \bot, &
\mc I (\varphi \star \psi) & := & \mc I(\varphi) \star \mc I(\psi), \\
\mc I (\forall x.\varphi(x)) & := & \bigwedge_{t \in \Tm} \mc I(\varphi(t)), &
\mc I (\exists x.\varphi(x)) & := & \bigvee_{t \in \Tm} \mc I(\varphi(t)),\\
\mc I (\forall X.\varphi) & := & \gamma ( |\forall X.\varphi|_{n+1}), &
\mc I (\exists X.\varphi) & := & |\exists X.\varphi|_{n+1}^\lhd,
\end{array}
$$
where $\star \in \{\wedge, \vee, \rightarrow\}$.
Notice our specific choice of level $n+1$ in the definitions of 
$\mc I (\forall X.\varphi)$ and 
$\mc I (\exists X.\varphi)$.
It can be flexibly changed, however, because of the following property.

\begin{lem}\label{l-down}
For every $0\leq k \leq n$,
$$\gamma(|\forall X.\varphi|_{n+1}) = \gamma(|\forall X.\varphi|_k),
\qquad
|\exists X.\varphi|_{n+1}^\lhd =
|\exists X.\varphi|_k^\lhd$$ 
hold if $\forall X.\varphi, \exists X.\varphi\in \FMP_k$.
\end{lem}

\begin{proof}
It is clear that the inclusion $\supseteq$ holds for the first equation.
So let $\Sigma \in |\forall X.\varphi|_{n+1}$ and 
$(\Gamma\Rightarrow\Pi) \in |\forall X.\varphi|_k^\rhd$.
The former means that $\Sigma \Rcf \varphi(Y)$ with $Y \not\in \FV(\Delta)$,
while the latter means that 
$\Delta,\Gamma\Rcf\Pi$ for any $\Delta \in 
|\forall X.\varphi|_k$. Hence we obtain
$\Sigma,\Gamma\Rcf\Pi$ by rule $(\tilde{\Omega}_k\Left)$.
Therefore $\Sigma \in |\forall X.\varphi|_k^{\rhd\lhd} = 
\gamma(|\forall X.\varphi|_k)$.

For the second equation, the inclusion $\subseteq$ holds because
$|\exists X.\varphi|_k \subseteq |\exists X.\varphi|_{n+1}$.
So let $\Sigma \in |\exists X.\varphi|_k^\lhd$ and 
$(\Gamma\Rightarrow\Pi) \in 
|\exists X.\varphi|_{n+1}$. The former means that 
$\Sigma, \Delta\Rcf\Lambda$ for any $(\Delta\Rightarrow\Lambda) \in 
|\exists X.\varphi|_k$, while the latter means that
$\varphi(Y),\Gamma\Rightarrow\Pi$ with $Y \not\in \FV(\Gamma,\Pi)$.
Hence we obtain
$\Sigma,\Gamma\Rcf\Pi$ by rule $(\tilde{\Omega}_k\Right)$.
Therefore $\Sigma \in 
|\exists X.\varphi|_{n+1}^\lhd$.
\end{proof}

Now a crucial lemma follows (called the ``main lemma''
in \cite{Okada02}).

\begin{lem}\label{l-main}
For every formula $\xi$ in $\FMP_n$, 
$\xi \in \mc I(\xi) \subseteq \xi^\lhd$.
\end{lem}

\begin{proof}
By induction on the structure of $\xi$.\\

\noindent
(1) $\xi$ is an atomic formula. We have $\mc I(\xi)=\gamma(\xi)$,
hence the claim holds by $(*)$ above.\\

\noindent
(2) $\xi = \varphi \wedge \psi$.
We first show $\varphi\wedge\psi \in 
\mc I(\varphi\wedge \psi)= \mc I(\varphi) \cap \mc I(\psi)$.
Let $(\Gamma\Rightarrow\Pi) \in \mc I(\varphi)^\rhd$. Then 
the induction hypothesis $\varphi \in \mc I(\varphi)$ yields 
$\varphi, \Gamma \Rightarrow^{cf} \Pi$, so 
$\varphi\wedge\psi, \Gamma \Rightarrow^{cf} \Pi$ by rule
$(\wedge\Left)$. That is, $\varphi\wedge\psi \in \mc I(\varphi)^{\rhd\lhd} 
= \mc I(\varphi)$.
Likewise, $\varphi\wedge\psi \in 
\mc I(\psi)$. So $\varphi\wedge\psi \in 
\mc I (\varphi \wedge \psi)$.

We next show 
$\mc I (\varphi \wedge \psi) \subseteq (\varphi \wedge \psi)^\lhd$.
Let $\Gamma \in \mc I (\varphi) \cap \mc I (\psi)$. Then we have 
$\Gamma \Rcf \varphi$ and $\Gamma \Rcf \psi$ by the induction hypotheses.
So $\Gamma \Rcf \varphi\wedge\psi$ by rule $(\wedge \Right)$.
That is, $\Gamma \in (\varphi\wedge\psi)^\lhd$.\\

\noindent
(3) $\xi = \varphi \vee \psi$.
We first show $\varphi\vee\psi \in 
\mc I(\varphi\vee \psi) = 
\gamma(\mc I(\varphi)\cup \mc I(\psi))$.
Let
$(\Gamma\Rightarrow\Pi) \in 
(\mc I(\varphi)\cup \mc I(\psi))^{\rhd}$.
Then the induction hypotheses 
$\varphi \in \mc I(\varphi)$ and 
$\psi \in \mc I(\psi)$ yield
$\varphi, \Gamma\Rcf \Pi$ and $\psi, \Gamma\Rcf \Pi$.
Hence $\varphi\vee\psi, \Gamma\Rcf\Pi$ by rule $(\vee \Left)$.
That is, $\varphi\vee\psi \in 
(\mc I(\varphi)\cup \mc I(\psi))^{\rhd\lhd}
= \mc I(\varphi \vee \psi)$.

We next show $\mc I (\varphi \vee \psi) \subseteq (\varphi \vee \psi)^\lhd$.
Let $\Gamma \in \mc I (\varphi) \cup \mc I (\psi)$, say $\Gamma \in 
\mc I (\varphi)$. Then 
$\Gamma \Rcf \varphi$ by the induction hypothesis. Hence 
$\Gamma \Rcf \varphi\vee\psi$ by rule $(\vee \Right)$.
That is, $\Gamma \in (\varphi\vee\psi)^{\lhd}$.
This proves that 
$\mc I (\varphi \vee \psi) \subseteq (\varphi\vee\psi)^{\lhd}$.\\

\noindent
(4) $\xi = \varphi \rightarrow \psi$.
We first show $\varphi\rightarrow\psi \in \mc I(\varphi\rightarrow \psi)$.
Let $\Sigma \in \mc I (\varphi)$ and $(\Gamma\Rightarrow\Pi) \in 
\mc I (\psi)^{\rhd}$.
Then $\Sigma \Rcf \varphi$ and $\psi,\Gamma\Rcf \Pi$ by the induction
hypotheses.
Hence $\Sigma, \varphi\rightarrow\psi,\Gamma \Rcf \Pi$ by rule
$(\rightarrow\Left)$. Since this holds for 
any $(\Gamma\Rightarrow\Pi) \in \mc I (\psi)^{\rhd}$, we have 
$\Sigma, \varphi\rightarrow\psi \in \mc I (\psi)^{\rhd\lhd}= \mc I (\psi)$.
Since this holds for any $\Sigma \in \mc I (\varphi)$, we conclude 
$\varphi\rightarrow\psi \in  \mc I(\varphi\rightarrow \psi)$.

We next show $\mc I (\varphi \rightarrow \psi) \subseteq 
(\varphi \rightarrow \psi)^\lhd$.
Let $\Gamma \in \mc I (\varphi) \rightarrow \mc I (\psi)$.
Since $\varphi \in \mc I(\varphi)$ and $\mc I(\psi) \subseteq \psi^\lhd$ 
by the induction hypotheses, we have $\varphi, \Gamma \Rcf \psi$. Hence
$\Gamma \Rcf \varphi\rightarrow\psi$ by rule $(\rightarrow\Right)$.
That is, $\Gamma \in (\varphi\rightarrow\psi)^\lhd$.\\

\noindent
(5) $\xi = \forall x.\varphi(x)$.
We first show $\forall x.\varphi(x) \in \mc I(\forall x.\varphi(x))
= \bigcap_{t\in\Tm}\mc I(\varphi(t))$.
Let $t \in \Tm$ and 
$(\Gamma\Rightarrow\Pi) \in \mc I(\varphi(t))^\rhd$. Then the induction 
hypothesis $\varphi(t) \in \mc I(\varphi(t))$ yields 
$\varphi(t), \Gamma \Rightarrow^{cf} \Pi$, so 
$\forall x.\varphi, \Gamma \Rcf \Pi$ by rule
$(\forall x\Left)$. That is, $\forall x.\varphi \in 
\mc I(\varphi(t))^{\rhd\lhd} = \mc I(\varphi(t))$.
Since it holds for any $t\in \Tm$, we conclude
$\forall x.\varphi \in \bigcap_{t\in \Tm} \mc I (\varphi(t))
 = \mc I (\forall x.\varphi(x))$.

We next show $\mc I(\forall x.\varphi(x)) \subseteq 
(\forall x.\varphi(x))^{\lhd}$.
Let $\Gamma \in \bigcap_{t\in \Tm} \mc I (\varphi(t))$. 
The proof splits into two cases. 
\begin{itemize}
\item[(i)] If $n= -1$, we choose a variable $y$ such that 
$y \not\in \Fv(\Gamma)$. Then $\Gamma \in
\mc I (\varphi(y))$. We have
$\Gamma \Rcf \varphi(y)$ by the induction hypothesis, so 
$\Gamma \Rcf \forall x.\varphi(x)$ by rule $(\forall x \Right)$.
That is, $\Gamma \in (\forall x.\varphi(x))^\lhd$.
\item[(ii)] If $n\geq 0$, we have $\Gamma \in \mc I (\varphi(t))
\subseteq \varphi(t)^\lhd$ for every $t\in \Tm$. Hence 
$\Gamma \Rcf\varphi(t)$, so $\Gamma\Rcf\forall x.\varphi(x)$ by 
rule $(\omega\Right)$. 
That is, $\Gamma \in (\forall x.\varphi(x))^\lhd$.
\end{itemize}

\noindent
(6) $\xi = \exists x.\varphi(x)$.
We first show $\exists x.\varphi(x) \in \mc I(\exists x.\varphi(x))$.
 Let 
$(\Gamma\Rightarrow\Pi) \in \left(\bigcup_{t\in \Tm} \mc I (\varphi(t))
\right)^\rhd$. The proof splits into two cases.
\begin{itemize}
\item If $n= -1$, choose a variable $y$ such that $y \not\in \Fv(\Gamma, \Pi)$.
Since $\varphi(y) \in \mc I(\varphi(y)) 
\subseteq \bigcup_{t\in \Tm} \mc I(\varphi(t))$ by the
induction hypothesis, we have
$\varphi(y), \Gamma\Rcf \Pi$.
Hence $\exists x.\varphi, \Gamma\Rcf\Pi$ by rule $(\exists x \Left)$.
That is, $\exists x.\varphi \in 
\left(\bigcup_{t\in \Tm} \mc I(\varphi(t))\right)^{\rhd\lhd} = 
\mc I(\exists x.\varphi(x))$.
\item If $n\geq 0$, we have
$\varphi(t) \in \mc I(\varphi(t)) 
\subseteq \bigcup_{t\in \Tm} \mc I(\varphi(t))$ for every $t\in \Tm$.
Hence 
$\varphi(t), \Gamma\Rcf \Pi$, so
$\exists x.\varphi, \Gamma\Rcf\Pi$ by rule $(\omega \Left)$.
That is, $\exists x.\varphi \in 
\mc I(\exists x.\varphi(x))$.
\end{itemize}

We next show $\mc I(\exists x.\varphi(x)) \subseteq 
(\exists x.\varphi(x))^{\lhd}$.
Let $\Gamma \in \bigcup_{t\in \Tm} \mc I(\varphi(t))$, 
say $\Gamma \in \mc I(\varphi(t))$. Then 
$\Gamma \Rcf \varphi(t)$ by the induction hypothesis. Hence 
$\Gamma \Rcf \exists x. \varphi$ by rule $(\exists x \Right)$.
That is, $\Gamma \in (\exists x.\varphi)^{\lhd}$.
This proves that 
$\mc I (\exists x.\varphi(x))\subseteq (\exists x.\varphi(x))^\lhd$.\\

\noindent
(7) $\xi = \bot$. Omitted.\\

\noindent
(8) $\xi = \forall X.\varphi$.
We have
$\forall X.\varphi \in |\forall X.\varphi|_{n+1} 
\subseteq \mc I(\forall X.\varphi)$, since 
$\forall X.\varphi\Rightarrow \varphi(Y)$ is cut-free provable in
$\LIO_n$. We also have
$|\forall X.\varphi|_{n+1}  \subseteq (\forall X.\varphi)^\lhd$
by rule $(\forall X\Right)$. Hence the claim follows.\\

\noindent
(9) $\xi = \exists X.\varphi$. Similarly to (8).
\end{proof}

\subsection{Algebraic cut elimination by $\Omega$-interpretation}
\label{ss-sound}

Given a sequent $\Gamma \Rightarrow \Pi$ in $\FMP_{n+1}$,
let 
$$
\begin{array}{llll}
\mc I(\Gamma) & := & \fwp (\Fm) & \mbox{if $\Gamma = \emptyset$},\\
 & := & \bigcap_{\varphi \in \Gamma} \mc I(\varphi) 
        & \mbox{otherwise},\\
\mc I(\Pi) & := & \gamma(\emptyset) & \mbox{if $\Pi = \emptyset$},\\
    & := & \mc I (\varphi) & \mbox{if $\Pi = \{\varphi\}$}.
\end{array}
$$
as in Section \ref{s-omegamacneille}.
We then have $\Gamma \in \mc I (\Gamma)$ and
$\mc I (\Pi) \subseteq \Pi^{\lhd}$ by Lemma \ref{l-main}.

Our next goal is to show that 
the interpretation $\mc I$ is sound for $\LIP_{n+1}$.
The proof consists of two steps.

Fix a set variable $X_0$ and $F \in 
\Gal(\m{CF})^\Tm$. Define interpretation 
$\mc I_F: \FMP_n \longrightarrow \Gal(\m{CF})$ 
similarly to $\mc I$, except that 
$\mc I_F (X_0(t)) := F(t)$. Notice that 
we have $\mc I_F(\varphi) = \mc I(\varphi)$ if
$X_0 \not\in \FV(\varphi)$.

\begin{lem}\label{l-sound1}
If $\LIO_n \vdash \Gamma \Rightarrow \Pi$, then 
$\mc I_F(\Gamma) \subseteq \mc I_F(\Pi)$.
\end{lem}

\begin{proof}
When $n=-1$, this follows from 
the standard soundness theorem for the first order intuitionistic logic.
So assume that $n\geq 0$. The proof proceeds by structural 
induction on the derivation. Let us only consider the cases for 
second order quantifiers.

\noindent
(1) The derivation ends with
$$
\infer[(\forall X \Right)]{\Gamma\Rightarrow \forall X.\varphi(X)}{
\Gamma\Rightarrow \varphi (Y) & Y \not\in \FV(\Gamma)}.$$
Let $\Delta \in \mc I_F(\Gamma)$. 
We may assume that $Y \neq X_0$ and $Y \not \in \FV(\Delta)$, since 
otherwise we can rename $Y$ to a new set variable.
By the induction hypothesis and 
Lemma \ref{l-main}, we have
$\Delta \in \mc I_F(\varphi(Y)) =
\mc I(\varphi(Y)) \subseteq \varphi(Y)^{\lhd}$, 
that is, $\Delta \Rcf \varphi(Y)$.
Hence $\Delta \in |\forall X.\varphi|_{n+1} \subseteq 
\mc I_F (\forall X.\varphi)$.\\

\noindent
(2) The derivation ends with
$$
\infer[(\exists X\Left)]{
\exists X.\varphi(X), \Gamma\Rightarrow \Pi}{
\varphi(Y),\Gamma \Rightarrow\Pi & Y\not\in \FV(\Gamma,\Pi)}.
$$
We assume $\Gamma=\emptyset$ for simplicity.
Let $(\Delta\Rightarrow\Lambda) \in \mc I_F(\Pi)^\rhd$. 
We may also assume that $Y \neq X_0$ and 
$Y \not\in \FV(\Delta,\Lambda)$.
By the induction hypothesis and Lemma \ref{l-main}, 
$\varphi(Y) \in \mc I(\varphi(Y)) \subseteq \mc I(\Pi)$,
so
$\varphi(Y), \Delta \Rcf \Lambda$. That is,
$\mc I_F(\Pi)^\rhd \subseteq |\exists X.\varphi|_{n+1}$.
Hence we conclude that $\mc I_F(\exists X.\varphi(X)) \subseteq
\mc I_F(\Pi)$.\\

\noindent
(3) The derivation ends with
$$
\infer[(\Omega_k \Left)]{\forall X.\varphi, \Gamma\Rightarrow\Pi}{
\{\;\Delta, \Gamma \Rightarrow \Pi\;\}_{\Delta\in 
|\forall X.\varphi|_k}}.
$$
We assume $\Gamma = \emptyset$ for simplicity.
We are going to use Lemma \ref{l-down}. 
So let $\Delta = \Delta(X_0) \in |\forall X.\varphi|_k$ and 
$(\Sigma\Rightarrow\Lambda) \in \mc I_F(\Pi)^\rhd$. 
The former means that $\LIO_{k-1} \vdash \Delta \Rightarrow \varphi(Y)$
for some $Y\not\in \FV(\Delta)$.
We have $\LIO_{k-1} \vdash \Delta(Z) \Rightarrow \varphi(Y)$ with $Z$ fresh
by Lemma \ref{l-subst}, so 
$\Delta(Z) \in |\forall X.\varphi|_k$.
Hence $\Delta(Z)\Rightarrow \Pi$ is among the 
premises. By the induction hypothesis,
$\mc I_F (\Delta(Z)) \subseteq \mc I_F(\Pi)$.
Since $\Delta(Z) \in \mc I(\Delta(Z)) = \mc I_F(\Delta(Z))$
by Lemma \ref{l-main}, we have
$\Delta(Z), \Sigma \Rcf \Lambda$. From this, 
$\Delta(X_0), \Sigma \Rcf \Lambda$ follows by Lemma \ref{l-subst}.
Since this holds for any 
$(\Sigma\Rightarrow\Lambda) \in \mc I_F(\Pi)^\rhd$, we have
$\Delta \in \mc I_F(\Pi)^{\rhd\lhd} = \mc I_F(\Pi)$. Finally we conclude
$$
\mc I_F(\forall X.\varphi) = \mc I(\forall X.\varphi) 
= \gamma(|\forall X.\varphi|_k) \subseteq \mc I_F(\Pi)
$$
by Lemma \ref{l-down}.\\

\noindent
(4) The derivation ends with
$$
\infer[(\Omega_k \Right)]{\Gamma\Rightarrow\exists X.\varphi}{
\{\; \Gamma,\Delta \Rightarrow \Lambda \;\}_{(\Delta\Rightarrow\Lambda)\in 
|\exists X.\varphi|_k}}.
$$
Again we are going to use Lemma \ref{l-down}.
So let $(\Delta\Rightarrow\Lambda) \in 
|\exists X.\varphi|_k$ and $\Sigma \in \mc I(\Gamma)$.
We may assume that $X_0 \not\in \FV(\Delta, \Lambda)$, since 
otherwise it can be renamed by a fresh variable as in (3) above.
By the induction hypothesis, 
$\mc I_F(\Gamma, \Delta) \subseteq \mc I_F(\Lambda)$.
Since $\Delta \in \mc I (\Delta) = \mc I_F(\Delta)$ and 
$\mc I_F(\Lambda) \subseteq \Lambda^\lhd$
by Lemma \ref{l-main}, we obtain
$\Sigma, \Delta \Rcf \Lambda$.
This shows that
$$\mc I_F(\Gamma) \subseteq 
|\exists X.\varphi|_{k}^\lhd =
\mc I(\exists X.\varphi) =
\mc I_F(\exists X.\varphi)$$
by Lemma \ref{l-down}.\\

\noindent
(5) The derivation ends with $(\tilde{\Omega}_k \Left)$ or
$(\tilde{\Omega}_k \Right)$. We do not have to consider these cases,
since these are derivable from other rules.
\end{proof}

\begin{lem}\label{l-sound2}
If $\LIP_{n+1} \vdash \Gamma\Rightarrow\Pi$, then 
$\mc I (\Gamma^\circ) \subseteq \mc I(\Pi^\circ)$ holds
for every term substitution $\circ$.
\end{lem}

\begin{proof}
The proof is very similar to that of the previous lemma.
The main difference is that we have to consider rules
$(\forall X\Left)$ and $(\exists X\Right)$ instead of 
$(\Omega_k\Left)$ and $(\Omega_k\Right)$. So let us deal with
only these two rules. For simplicity, we assume that $\circ$ is 
an identity substitution.\\

\noindent
(1) The derivation ends with rule $(\forall X\Left)$.
It is sufficient to show that $\mc I(\forall X.\varphi(X)) \subseteq
\mc I(\varphi(\tau))$ for any $\forall X.\varphi(X) \in \FMP_{n+1}$ and 
$\tau \in \ABS_{n+1}$. 
Let $\Delta \in |\forall X.\varphi|_{n+1}$, that is,
$\Delta \Rcf \varphi(X_0)$ with $X_0$ fresh.
Define $F \in \Gal(\m{CF})^{\Tm}$ by 
$F(t) := \mc I(\tau(t))$.
By Lemma \ref{l-sound1}, 
we have $\mc I_F(\Delta) \subseteq \mc I_F(\varphi(X_0))$.
Notice that:
\begin{itemize}
\item $\mc I_F(\Delta) = \mc I(\Delta)$ because $X_0\not\in \FV(\Delta)$.
\item $\Delta \in \mc I(\Delta)$ by Lemma \ref{l-main}, noting that 
$\Delta \subseteq \FMP_n$.
\item $\mc I_F(\varphi(X_0)) = \mc I(\varphi(\tau))$ by induction 
on the structure of $\varphi(X_0)$.
\end{itemize}
Hence we have $\Delta \in \mc I(\varphi(\tau))$.
This shows that
$\mc I(\forall X.\varphi(X)) = \gamma(|\forall X.\varphi|_{n+1})
\subseteq\mc I(\varphi(\tau))$.\\

\noindent
(2) The derivation ends with rule $(\exists X\Right)$.
It is sufficient to show that $\mc I(\varphi(\tau)) \subseteq
\mc I(\exists X.\varphi(X))$ 
for any $\exists X.\varphi(X)\in \FMP_{n+1}$ and $\tau \in \ABS_{n+1}$. 
Let $(\Delta\Rightarrow\Lambda) \in |\exists X.\varphi|_{n+1}$, that is,
$\varphi(X_0), \Delta \Rcf \Lambda$ with $X_0$ fresh.
Define $F \in \Gal(\m{CF})^{\Tm}$ by 
$F(t) := \mc I(\tau(t))$.
By Lemma \ref{l-sound1}, $\mc I_F(\varphi(X_0),
\Delta) \subseteq \mc I_F(\Lambda)$.
As before, we have:
\begin{itemize}
\item $\mc I_F(\Delta) = \mc I(\Delta)$ and 
$\mc I_F(\Lambda) = \mc I(\Lambda)$. 
\item $\Delta \in \mc I(\Delta)$ and $\mc I(\Lambda) \subseteq \Lambda^{\lhd}$.
\item $\mc I_F(\varphi(Y)) = \mc I(\varphi(\tau))$.
\end{itemize}
Altoghether, this means that $\Sigma, \Delta \Rcf \Lambda$
holds for any $\Sigma \in \mc I(\varphi(\tau))$
and any $(\Delta\Rightarrow\Lambda) \in |\exists X.\varphi|_{n+1}$.
Hence we conclude that 
$\mc I(\varphi(\tau)) \subseteq 
|\exists X.\varphi|_{n+1}^\lhd =
\mc I(\exists X.\varphi(X))$.
\end{proof}

By combining Lemmas \ref{l-main} and \ref{l-sound2}, we obtain
an algebraic proof of (partial) cut elimination.

\begin{thm}\label{t-reduction}
$\LIP_{n+1} \vdash \Gamma\Rightarrow\Pi$ implies
$\LIO_{n} \vdash^{cf} \Gamma\Rightarrow\Pi$ provided that 
$\Gamma\cup\Pi \subseteq \FMP_n$.
\end{thm}

\begin{proof}
We have $\mc I(\Gamma) \subseteq \mc I(\Pi)$
by Lemma \ref{l-sound2}, while
$\Gamma \in \mc I(\Gamma)$ and $\mc I(\Pi) \subseteq \Pi^\lhd$
by Lemma \ref{l-main}.
Hence $\Gamma \Rcf \Pi$, that is, $\LIO_n \vdash \Gamma\Rightarrow \Pi$.
\end{proof}

Together with Lemma \ref{l-collapse}, it leads to:
\begin{thm}[Algebraic cut elimination for $\LIP_{n+1}$]\label{t-last}
For every $n\geq -1$,
$\LIP_{n+1} \vdash \Gamma\Rightarrow\Pi$ implies
$\LI \vdash^{cf} \Gamma\Rightarrow\Pi$ provided that 
$\Gamma\cup\Pi \subseteq \Fm$.
\end{thm}

By combining this with Theorem \ref{t-partial2}, we may obtain 
complete cut elimination too.

\begin{cor}
For every $n\geq -1$ and every sequent $\Gamma\Rightarrow\Pi$ of 
$\LIP_{n+1}$,
$\LIP_{n+1} \vdash \Gamma\Rightarrow\Pi$ implies
$\LIP_{n+1} \vdash^{cf} \Gamma\Rightarrow\Pi$.
\end{cor}


\subsection{A small modification}\label{ss-modification}

Theorem \ref{t-reduction} gives a self-contained proof of 
the reduction from $\LIP_{n+1}$ to \emph{cut-free} $\LIO_n$.
On the other hand, it is also possible to prove the same thing 
relying on the fact that $\LIO_n$ admits cut elimination (Lemma \ref{l-cut}).
This approach provides a simpler interpretation for 
the formulas in $\FMP_n$ that will play an important
technical role in the next section.
It will also help clarify the essence of our algebraic argument so far.

\begin{lem}\label{l-improvement}
Assume that $\LIO_n$ admits cut elimination. Then 
$$\mc I(\varphi) = \gamma(\varphi)=  \varphi^{\lhd}$$
holds for every $\varphi \in \FMP_n$. As a consequence,
$$
\gamma(\varphi\star\psi) = \gamma(\varphi)\star\gamma(\psi),\qquad
\gamma(\forall x.\varphi(x)) = \bigcap_{t\in \Tm} \gamma(\varphi(t)),\qquad
\gamma(\exists x.\varphi(x)) = 
\gamma\left(\bigcup_{t\in \Tm} \gamma(\varphi(t))\right)
$$
hold for every $\star \in \{\wedge, \vee, \rightarrow\}$ and 
$\varphi, \psi \in \FMP_n$.
\end{lem}

\begin{proof}
By Lemma \ref{l-main}, we have
$\gamma(\varphi) \subseteq \mc I(\varphi) \subseteq \varphi^{\lhd}$.
Hence it suffices to show that $\varphi^{\lhd} \subseteq \gamma(\varphi)$.
Let $\Sigma \in \varphi^\lhd$ and 
$(\Delta\Rightarrow\Lambda) \in \varphi^{\rhd}$. Then 
we have $\Sigma \Rcf \varphi$ and $\varphi, \Delta\Rcf \Lambda$,
so $\Sigma, \Delta \Rcf \Lambda$ by cut elimination.
Since this holds for any 
$(\Delta\Rightarrow\Lambda) \in \varphi^{\rhd}$, we conclude 
$\Sigma \in \gamma(\varphi)$.
\end{proof}

Lemma \ref{l-improvement} gives us a connection 
with the Lindenbaum algebra.
Let us assume $n=-1$ for simplicity. Because of the lemma,
we may restrict the underlying set $\Gal(\m{CF})$ of 
the Heyting algebra $\m{CF}^+$ to 
$\{ \gamma(\varphi) : \varphi \in \Fm\}$. This results in
a Heyting subalgebra $\m{CF}^+_0$.
Notice that 
$$
\gamma(\varphi)\subseteq \gamma(\psi)
\Iff \LI\vdash \varphi \Rightarrow \psi.
$$
Hence $\m{CF}^+_0$ is isomorphic to
the Lindenbaum algebra for $\LI$
 ($\m L$ in Section \ref{s-omegamacneille}). 
Moreover it is not hard to see that
$\m{CF}^+$ is the MacNeille completion of $\m{CF}^+_0$. To sum up:
\begin{prop}
If $n=-1$, 
$\m{CF}^+$ is isomorphic to the MacNeille completion of 
the Lindenbaum algebra for $\LI$.
\end{prop}

Hence it seems reasonable to describe the essence of our arugment
as 
$$
\mbox{MacNeille completion} + \mbox{$\Omega$-interpretation}.
$$
This can be compared with the essence of the algebraic proof of
cut elimination for $\LIT$ due to \cite{Maehara91,Okada02}:
$$
\mbox{MacNeille completion} + \mbox{reducibility candidates}.
$$
See Appendix \ref{a-okada} for the latter.

\section{Formalizing cut elimination}\label{s-formalize}

In this section
we outline how to formalize our proof of cut
elimination for $\LIP_{n+1}$ locally within $\ID_{n+1}^i$.
It consists of two steps (Subsections \ref{ss-form1} and \ref{ss-form2}).

\subsection{Formalization in $\ID_{n+1}^i$ (1)}\label{ss-form1}

Recall that
the syntactic proof of cut elimination for 
$\LIP_n$ (in Subsection \ref{ss-syntactic})
relies on a cut-free provability predicate
$\LIO_{k}^{cf}$ for $k=-1, \dots, n$, which are definable in $\ID_{n+1}^i$.
Concretely, there is an $\ID_{n+1}$-formula $\LLIO_{k}^{cf}(x)$ 
whose intended meaning is
$$
\LLIO_k^{cf}(\Code{\Gamma\Rightarrow\Pi}) 
\Iff \LIO_{k} \vdash^{cf} \Gamma\Rightarrow\Pi.
$$
Likewise, we are given formulas $\LLIO_k(x)$, 
$\LLIP_{n+1}(x)$ and $\LLI^{cf}(x)$ in $\ID_{n+1}$
with analogous meanings.
Notice that  $\LLIP_{n+1}(x)$ and $\LLI^{cf}(x)$ are actually
$\Sigma^0_1$ formulas of $\HA$, since the proof systems 
are finitary. 
We assume that Lemma \ref{l-cut} for $\LIO_n$ has been
already formalized in $\ID_{n+1}^i$ (that is part of Theorem \ref{t-ce}). 
Thus:
$$\ID_{n+1}^i \vdash \forall x \in \SSeq.\,
\LLIO_{n}(x)\rightarrow \LLIO_{n}^{cf}(x),
$$
where $\SSeq(x)$ is a unary predicate for the set 
of (codes of) $\LIO_n$-sequents.

Let us now turn to the algebraic side.
Recall that our syntactic frame $\m{CF}$ is defined in terms of 
cut-free provability in $\LIO_n$. Thus the closure operator
$\gamma$ can be formalized as follows.
Given a set variable $X$, let 
$$
\Cl(x, X) \Def \forall y\in \SSeq.\,
[\forall z \in X.\, \LLIO_n^{cf} 
(\Code{\dot{z} \Rightarrow \dot{y}})] \rightarrow 
\LLIO_n^{cf} 
(\Code{\dot{x} \Rightarrow \dot{y}}),
$$
where $\Code{\dot{x} \Rightarrow \dot{y}}$ is a function symbol 
in variables $x,y$, whose intended meaning is a function
that returns $\Code{\Sigma,\Gamma\Rightarrow\Pi}$
given $\Code{\Sigma}$ and $\Code{\Gamma\Rightarrow\Pi}$ as inputs.
The intended meaning of $\Cl$ is that 
$\Cl(\Code{\Delta}, X)$ iff $\Delta \in \gamma(X)$.

Based on this, we can define an
$\ID_{n+1}$-formula $\II_\varphi (x,\vec{y})$ 
for each $\varphi = \varphi(\vec{y}) \in \FMP_{n+1}$
such that
$$
\II_\varphi (\Code{\Delta}, \Code{\vec{t}}) 
\Iff \Delta \in \mc I(\varphi(\vec{t})).
$$
This is possible because of the $\Omega$-interpretation technique,
that allows us to interpret second order quantifiers by 
first order ones. It is not very hard to 
formalize Lemma \ref{l-main}:
$$
\begin{array}{l}
\xi(\vec{t}) \in \mc I(\xi(\vec{t})),\\
\Delta \in \mc I(\xi(\vec{t})) \quad\mbox{implies}\quad
\Delta\Rcf\xi(\vec{t}),
\end{array}
$$
for every $\xi(\vec{x}) \in \FMP_n$, 
$\vec{t} \in \Tm$ and $\Delta \fsubseteq \FMP_n$. 

\begin{lem}\label{l-mainform}
For every formula $\xi = \xi(\vec{y})$ in $\FMP_n$, $\ID_{n+1}^i$ proves
$$
\begin{array}{l}
\forall \vec{y} \in \TTm.\,
\II_{\xi} (\Code{\xi(\vec{\dot{y}})}, \vec{y}),\\
\forall \vec{y} \in \TTm.\,
\forall z \in \FFset.\,
\II_{\xi} (z, \vec{y}) \rightarrow \LLIO_n^{cf}(
\Code{\dot{z} \Rightarrow \xi(\vec{\dot{y}})}).
\end{array}$$
\end{lem}

Here $\TTm(x)$ and $\FFset(x)$ are unary predicates for 
the set of (codes of) terms and the set of (codes of) finite formula sets
$\Gamma, \Delta, \dots$. In addition,
$\Code{\xi(\vec{\dot{y}})}$ expresses a function in variables $\vec{y}$ that 
returns $\Code{\xi(\vec{t})}$ when given $\Code{\vec{t}}$ as inputs.
$\Code{\dot{z} \Rightarrow \xi(\vec{\dot{y}})}$ should be understood
accordingly.

Now the backbone of our argument is Lemma \ref{l-sound2}.
So suppose that a derivation $\pi_0$ of $\Gamma_0\Rightarrow\Pi_0$
in $\LIP_{n+1}$ is given, where $\Gamma_0 \cup \Pi_0 \subseteq \FMP_n$. 
Since there are only finitely many formulas occurring in $\pi_0$,
we obtain a formula $\II (x, y, \vec{z})$ such that 
$$
\II (\Code{\Delta}, \Code{\Gamma}, \Code{\vec{t}}) 
\Iff \Delta \in \mc I(\Gamma(\vec{t})))
$$
for any $\Delta \fsubseteq \FMP_{n}$ and 
$\Gamma = \Gamma(\vec{y})$ that consists of 
formulas occurring in $\pi_0$.

We would like to show that $\ID_{n+1}^i$ proves 
that $\mc I (\Gamma_0) \subseteq \mc I(\Pi_0)$
(formally expressed by using predicate $\II$ above).
If successful,
we may further obtain 
$$\ID_{n+1}^i \vdash \LLIO^{cf}_{n} (\Code{\Gamma_0\Rightarrow\Pi_0})$$
with the help of Lemma \ref{l-mainform} (see the proof of 
Theorem \ref{t-reduction}).
Combined with a formalized proof of Lemma \ref{l-collapse},
we will be able to conclude
$$\ID_{n+1}^i \vdash \LLI^{cf} (\Code{\Gamma_0\Rightarrow\Pi_0}).$$

The hardest part of the whole work is to property formalize 
Lemma \ref{l-sound1}, which is a prerequisite for 
Lemma \ref{l-sound2}.
We will argue that it is indeed possible
in the next subsection.

\subsection{Local formalization in $\ID_{n+1}^i$ (2)}\label{ss-form2}

Before addressing Lemma \ref{l-sound1}, a bit of preliminary 
is needed. 

We fix a variable $X_0$, a formula $\forall X.\varphi_0(X)$ and 
abstract $\tau_0$ that occur in 
the derivation $\pi_0$ of $\Gamma_0\Rightarrow\Pi_0$ in
$\LIP_{n+1}$.
Lemma \ref{l-sound1} is invoked by letting 
$F:= \mc I(\tau_0)$ and by considering 
a cut-free derivation of $\Delta \Rightarrow \varphi_0(X_0)$ 
in $\LIO_n$ with $X_0 \not\in \FV(\Delta)$
(see case (1) in the proof of Lemma \ref{l-sound2}).
Actually, there is also a dual case corresponding to case (2), 
but let us forget about it for simplicity.
So, there is an $\ID_{n+1}$-formula 
$\mathsf{F}(x,y)$ whose intended meaning is that 
$$\mathsf{F}(\Code{\Delta}, \Code{t}) \Iff
\Delta \in F(t) \Iff \Delta \in \mc I (\tau (t)).$$

We define the \emph{subformula relation} to be the transitive reflexive
closure of the following:
$$
\varphi, \psi \sqsubseteq \varphi\star\psi,\qquad
\varphi(x) \sqsubseteq Qx. \varphi(x),
$$
where $\star\in \{\wedge,\vee,\rightarrow\}$ and $Q\in \{\forall,\exists\}$.
That is, $QX. \varphi(X)$ does not have any proper subformula.
Clearly the set $\Sf(\varphi_0)$ of subformulas of $\varphi_0$
is finite. Hence as before, there is a formula 
$\IF(x,y,\vec{z})$ such that
$$
\IF (\Code{\Delta}, \Code{\Gamma}, \Code{\vec{t}}) 
\Iff \Delta \in \mc I_F(\Gamma(\vec{t})))
$$
for any $\Delta\fsubseteq \FMP_{n}$ and 
$\Gamma = \Gamma(\vec{y}) \fsubseteq \Sf(\varphi_0)$.

Let $\SEQ(\varphi_0)$ be the (finite) set of sequents 
that consist of formulas in $\Sf(\varphi_0)$.
Then Lemma \ref{l-sound1} can be formalized as follows.

\begin{lem}
$\ID_{n+1}^i$ proves the following statement (suitably formalized):
for any sequent $\Gamma(\vec{y}) \Rightarrow \Pi(\vec{y})$ in 
$\SEQ(\varphi_0)$, for any terms $\vec{t}$ and 
for any $\Delta\fsubseteq \FMP_{n}$ with $X_0 \not\in\FV(\Delta)$,
$$
\LIO_n \vdash^{cf} \Delta, \Gamma(\vec{t})\Rightarrow\Pi(\vec{t})
\quad \mbox{implies}\quad
\gamma(\Delta)\cap \mc I_F(\Gamma(\vec{t})) \subseteq \mc I_F(\Pi
(\vec{t})).$$
Also, for any 
$\Gamma(\vec{y}) \fsubseteq \Sf(\varphi_0)$, 
for any terms $\vec{t}$ 
and for any sequent $\Delta\Rightarrow\Pi$ with 
$\Delta\cup\Pi \subseteq \FMP_{n}$ and $X_0\not\in\FV(\Delta,\Pi)$,
$$
\LIO_n \vdash^{cf} \Gamma(\vec{t}), \Delta \Rightarrow \Lambda
\quad \mbox{implies}\quad
I_F(\Gamma(\vec{t})) \cap \gamma(\Delta)\subseteq \Lambda^\lhd.$$
\end{lem}

\begin{proof} 
By structural induction on the cut-free derivation
(that is available in $\ID_{n+1}^i$).
Since $\gamma(\Delta) = \mc I_F(\Delta)$ and $\Lambda^{\lhd} = \mc I_F(\Lambda)$
by Lemma \ref{l-improvement},
the first statement amounts to:
$$
\LIO_n \vdash^{cf} \Delta, \Gamma(\vec{t})\Rightarrow\Pi(\vec{t})
\quad \mbox{implies}\quad
\mc I_F(\Delta,\Gamma(\vec{t})) \subseteq \mc I_F(\Pi(\vec{t})).$$
Hence the proof of Lemma \ref{l-sound1} can be 
formalized almost straightforwardly.
\end{proof}

We also have to ensure that the whole construction is primitive recursive.

\begin{lem}
Given $\forall X.\varphi_0(X) \in \FMP_{n+1}$ and $\tau_0 \in \ABS_{n+1}$,
there is a derivation $\pi$ of the following statement 
(formalized in $\ID_{n+1}^i$):
for any $\Delta\fsubseteq \FMP_{n}$ with $X_0 \not\in\FV(\Delta)$,
$$
\LIO_n \vdash^{cf} \Delta \Rightarrow \varphi_0(X_0)
\quad \mbox{implies}\quad
\Delta \in \mc I_F(\varphi_0(X_0)).$$
Moreover, $\pi$ is computable from 
$\forall X.\varphi_0(X)$ and $\tau$ primitive recursively.
\end{lem}

This is certainly true since all the reasoning is constructive 
and parametric in 
$\forall X.\varphi_0(X)$ and $\tau_0$.

Let us now come back to the derivation $\pi_0$
of $\Gamma_0 \Rightarrow\Pi_0$ in $\LIP_{n+1}$.
In the proof of Lemma \ref{l-sound2},
Lemma \ref{l-sound1} is invoked finitely many times
depending on $\pi_0$. Moreover,
we can verity that $\ID_{n+1}^i$ proves
$$
\mc I_F(\varphi_0(Y)) = \mc I(\varphi_0(\tau_0)).
$$
Hence Lemma \ref{l-sound2} can be formalized as follows.
\begin{lem}\label{l-sound2form}
Let $\pi_0$ be a derivation of 
$\Gamma_0 \Rightarrow\Pi_0$ in $\LIP_{n+1}$.
There is a derivation $\pi_1$ of the statement 
$\mc I (\Gamma_0) \subseteq \mc I(\Pi_0)$ 
formalized in $\ID_{n+1}^i$. Moreover, $\pi_1$ is computable
from $\pi_0$ primitive recursively.
\end{lem}

This is again a matter of routine work.

As explained before, this lemma together with Lemma \ref{l-mainform}
and (a formalized version of) Lemma \ref{l-collapse} gives rise to a proof
of partial cut elimination for $\LIP_{n+1}$ locally formalized in $\ID_{n+1}^i$.
Let us record this fact (with $m:=n+1$).

\begin{thm}
$\IS$ proves the statement 
that for every sequent $\Gamma\Rightarrow\Pi$ of $\LI$,
$$\LIP_m \vdash\Gamma\Rightarrow\Pi\quad
\mbox{implies}\quad \ID_m^i \vdash \LLI^{cf}(\Code{\Gamma\Rightarrow\Pi}).$$
\end{thm}

Assuming the 1-consistency of $\ID_m^i$, we obtain
$$\LIP_m \vdash\Gamma\Rightarrow\Pi\quad
\mbox{implies}\quad \LI\vdash^{cf}\Gamma\Rightarrow\Pi,$$
that is nothing but a statement of partial cut elimination.
Hence
by combining it with Theorems \ref{t-cc2} and \ref{t-partial2}, 
we finally obtain:

\begin{thm}[Takeuti correspondence between $\ID_m^i$ and 
$\LIP_{m}$]\label{t-final}
For every $m<\omega$,
$$
\IS \vdash \mathsf{CE}(\LIP_m) \leftrightarrow 
\mathsf{1CON}(\ID_m^i).$$
\end{thm}

In the above theorem, $\ID_m^i$ can be replaced by $\ID_m$
by Theorem \ref{t-ci}. Also, $\LIP_m$ can be replaced by
its classical counterpart since our proof of
cut elimination works for classical systems with some minor changes.

\begin{rem}
It is not our original idea to 
combine a syntactic argument based on the $\Omega$-rule 
with a semantic argument 
to save one inductive definition. For instance,
Aehlig \cite{Aehlig05} employs Tait's computability predicate
defined on a provability predicate based on the $\Omega$-rule.
He works on the parameter-free, \emph{negative} fragments of
second order Heyting arithmetic without induction,
and proves partial cut elimination in the corresponding $\ID$-theories.
His result is comparable with ours, but our approach 
based on the MacNeille completion works 
for logical systems with the \emph{full} set of connectives 
(recall that second order definitions of positive connectives $\{\vee,\exists\}$
are not available in the parameter-free setting). Moreover,
it works for \emph{classical} logical systems too
(because the variety of Boolean algebras is closed under 
MacNeille completions).
\end{rem}

\section{Conclusion}

In this paper, we have brought 
the $\Omega$-rule technique originally developed in arithmetic 
into the logical setting, and studied it from 
an algebraic perspective. We have found an intimate connection with
the MacNeille completion (Theorem \ref{t-crucial}), 
that is important in two ways.
First, it provides an unexpected link between ordered algebra
and proof theory. Second, it inspires 
an algebraic form of the $\Omega$-rule, called the $\Omega$-interpretation,
that can be used to give 
an algebraic proof of cut elimination for $\LIP_m$ (with $m <\omega$).
As we have argued in Subsection \ref{ss-modification},
the essence of our approach could be summarized as
$$
\mbox{MacNeille completion} + \mbox{$\Omega$-interpretation}.
$$
This combination, together with some syntactic arguments, 
leads to a cut elimination proof 
which is locally formalizable in $\ID_m^i$.

An outcome is the Takeuti correspondence between $\ID$-theories
and parameter-free logics
(Theorem \ref{t-final}):
$$
\IS \vdash \mathsf{CE}(\LIP_m) \leftrightarrow 
\mathsf{1CON}(\ID_m^i).$$
This result should not be surprising 
for proof theorists at all,
although we do not find any work 
formally proving this in the literature 
(either for the intuitionistic or classical logic).

Our emphasis rather lies in the methodological aspect. The algebraic 
approach works fine not just for full second order logics
but also for their parameter-free fragments. Moreover, it works
uniformly both for the intuitionistic and classical logics because 
of a purely algebraic reason:
the variety of Heyting algebras and that of Boolean algebras
are both closed under MacNeille completions (Theorem \ref{t-limit}).

Our intuitionistic sequent calculus $\LIP_{<\omega}$ roughly 
corresponds to the classical calculus studied in \cite{Takeuti58}.
Hence what we have achieved in this paper
is to algebraically reformulate Takeuti's classical cut elimination 
theorem that accounts for the 1-consistency of $\PICA$
\cite{Takeuti58}.
Our hope is to expand this algebraic approach to 
more recent advanced results in proof theory, although we are 
not optimistic at all.

\section*{Acknowledgment}
The author is grateful to Ryota Akiyoshi for useful comments.

\appendix

\section{Proof of Theorem \ref{t-cc}}\label{a-proof1}

The purpose of this section is to prove Lemma \ref{l-cc3},
which is used in the proof of Theorem \ref{t-cc}.

First of all, recall that the axioms of $\PA$ (and $\HA$) consist of 
the equality axioms, $\forall xy.\, s(x)=s(y)\rightarrow x=y$,
$\forall x. s(x) \neq 0$,
the defining axioms for primitive recursive functions
as well as the induction axioms. 
All are $\Pi^0_1$ sentences except the last ones.


\begin{lem}\label{l-cc0}
Given a term $t$, let $\Nn_t(x) := \Nn(t(x))$.
$\LIP_0$ proves
$$
[\forall x \in{\Nn}.\, \Nn_t(x)\rightarrow\Nn_t(s(x))] \wedge
\Nn_t(0) \rightarrow  \forall y\in \Nn. \Nn_t(y).
$$
\end{lem}

This can be proved similarly to Lemma \ref{l-ind}.

Given a list $\vec{x} = x_1, \dots, x_n$
of variables, we denote 
the list $\Nn(x_1), \dots, \Nn(x_n)$ by $\Nn(\vec{x})$.
In the following lemmas, when we write $t(\vec{x})$ or $\varphi(\vec{x})$,
we assume that all free variavles in $t$ or $\varphi$ are 
included in $\vec{x}$.

\begin{lem}\label{l-cc1}
For every term $t(\vec{x})$ over $L_{\PA}$, $\LIP_0$ proves
$$
\Nn(\vec{x}), \Gamma \Rightarrow \Nn(t(\vec{x})),
$$
where $\Gamma$ consists of some
$\Pi^0_1$ axioms of $\PA$.
\end{lem}

\begin{proof}
We have already seen that $\LIP_0$ proves
$\Nn(0)$ and $\Nn(x)\Rightarrow \Nn(s(x))$. Hence
it is sufficient to show that 
$\LIP_0 \vdash \Nn(\vec{x}) \Rightarrow \Nn(f(\vec{x}))$ 
for each symbol $f$ for a primitive recursive function. 
We only consider a simplified case, where
$f$ is 
defined from constant $c$ and binary function $h$ by the sentence:
$$
\boldsymbol{Def}(f) \Def f(0)=c \wedge \forall x. f(s(x))=h(x, f(x)).
$$
We further assume that the claim has been proved for $c$ and $h$. 
That is, $\LIP_0$ proves
$$
\Gamma \Rightarrow \Nn(c), \qquad  \Nn(x), \Nn(y),\Gamma\Rightarrow \Nn(h(x,y)).
$$
From the former, we obtain
$$\Gamma, \boldsymbol{Def}(f)\Rightarrow \Nn(f(0))$$
by using $Sub(\Nn)$.
From the latter, we obtain
$\Gamma, \boldsymbol{Def}(f), \Nn(x), \Nn(f(x)) \Rightarrow \Nn(f(s(x)))$, so
$$\Gamma, \boldsymbol{Def}(f) \Rightarrow \forall x\in \Nn.\, 
\Nn(f(x)) \rightarrow \Nn(f(s(x))).$$ 
Hence by Lemma \ref{l-cc0}
we obtain 
$\Gamma, \boldsymbol{Def}(f) \Rightarrow \forall y\in\Nn. \Nn(f(y))$.
Therefore,
$$\LIP_0 \vdash \Nn(y), \Gamma, \boldsymbol{Def}(f) \Rightarrow \Nn(f(y))$$
 as required.
\end{proof}

Once Lemma \ref{l-cc1} has been proved, it is routine 
to prove the relativization lemma below.

\begin{lem}[Relativization]\label{l-cc2}
$\LI\vdash \Gamma\Rightarrow\Pi$ implies 
$\LIP_0 \vdash \Nn(\vec{x}), \Gamma^{\Nn} \Rightarrow \Pi^{\Nn}$,
where $\Fv(\Gamma,\Pi) \subseteq \{\vec{x}\}$.
\end{lem}

Now let us put things together.

\begin{lem}\label{l-cc3}
If $\HA$ proves a $\Sigma^0_1$ sentence $\varphi$,
then $\LIP_0$ proves $\Gamma \Rightarrow \varphi$ 
where $\Gamma$ consists of some $\Pi^0_1$ axioms of $\PA$.
\end{lem}

\begin{proof}
If $\HA$ proves $\varphi$,
$\LI$ proves
$\Gamma, \Delta \Rightarrow \varphi$
where $\Gamma$ consists of $\Pi^0_1$ axioms and 
$\Delta$ of induction axioms.
By Lemmas \ref{l-ind} and \ref{l-cc2}, we obtain
$\LIP_0 \vdash \Gamma^{\Nn} \Rightarrow \varphi^{\Nn}$.
Since each $\psi\in \Gamma$ is $\Pi^0_1$ and $\varphi$ is $\Sigma^0_1$,
we have
$
\psi \Rightarrow \psi^{\Nn}$ and $\varphi^{\Nn} \Rightarrow \varphi$,
so that we finally obtain
$\LIP_0 \vdash \Gamma \Rightarrow\varphi$.
\end{proof}


\section{Proof of Lemma \ref{l-cut}}\label{a-proof2}

First, we define the \emph{rank} of 
each formula $\varphi \in \FMP_n$, denoted by 
$\Rk (\varphi)$, as follows:
\begin{itemize}
\item $\Rk(\bot) = \Rk(X(t)) = \Rk(p(\Vec{t})) = \Rk(\forall X.\xi) 
= \Rk(\exists X.\xi) := 0$,
\item $\Rk(\varphi \star \psi) := 
\max \{\Rk(\varphi), \Rk(\psi)\} +1$ ($\star \in \{\wedge,\vee,\rightarrow\}$),
\item $\Rk(\forall x. \varphi) = \Rk(\exists x. \varphi) := \Rk(\varphi)+1$.
\end{itemize}

Given an ordinal $\alpha \leq \omega$,
we write $\vdash_\alpha \Gamma\Rightarrow\Pi$
if $\Gamma\Rightarrow\Pi$ has a derivation in $\LIO_n$
in which all cut formulas 
are of rank strictly less than $\alpha$.
Thus $\vdash_\omega \Gamma\Rightarrow\Pi$ means 
$\LIO_n \vdash \Gamma\Rightarrow\Pi$, and 
$\vdash_0 \Gamma\Rightarrow\Pi$ means 
$\LIO_n \vdash^{cf} \Gamma\Rightarrow\Pi$.

\begin{lem}\label{l-cut0}
Let $m<\omega$. 
Suppose that $\Rk(\varphi) \leq m$, 
$\vdash_m \varphi,\Gamma\Rightarrow\Pi$ and 
$\vdash_m \Gamma\Rightarrow\varphi$, where
in the derivation of $\Gamma\Rightarrow\varphi$
the main formula of 
the last inference step is the indicated $\varphi$.
Then
$\vdash_m \Gamma\Rightarrow \Pi$.
\end{lem}

\begin{proof}
Let $\pi_l$ be the derivation of $\Gamma\Rightarrow\varphi$
and $\pi_r$ that of $\varphi,\Gamma\Rightarrow\Pi$. We argue
by structural induction on $\pi_r$.
Let us only verify a few cases.\\

\noindent
(1) The main formula of the last inference of $\pi_r$ is not $\varphi$.
In this case, the claim follows immediately from the induction hypothesis.\\

\noindent
(2) $\pi_l$ and $\pi_r$ respectively end with
$$
\infer[(\rightarrow \Right)]{\Gamma \Rightarrow \varphi_1\rightarrow\varphi_2}{
\Gamma, \varphi_1 \Rightarrow \varphi_2},
\qquad
\infer[(\rightarrow \Left)]{
\varphi_1 \rightarrow \varphi_2, \Gamma\Rightarrow\Pi}{
\varphi_1 \rightarrow \varphi_2, \Gamma\Rightarrow\varphi_1 & 
\varphi_2,\varphi_1 \rightarrow \varphi_2, \Gamma\Rightarrow\Pi},
$$
where we assume that the upper sequents of rule $(\rightarrow\Left)$ 
contain $\varphi_1 \rightarrow \varphi_2$ in the antecedent.
The argument would be simpler if the formula is absent.

By the induction hypothesis, we have
$\vdash_m \Gamma \Rightarrow \varphi_1$
and 
$\vdash_m \varphi_2, \Gamma\Rightarrow\Pi$. Hence 
applying $(\Cut)$ on $\varphi_1$ and $\varphi_2$, we obtain
$\vdash_m \Gamma\Rightarrow \Pi$, noting that 
$\Rk(\varphi_i) < \Rk(\varphi_1 \rightarrow\varphi_2)\leq m$ for $i=1,2$.\\

\noindent
(3) $\pi_l$ and $\pi_r$ respectively end with
$$
\infer[(\omega\Right)]{\Gamma\Rightarrow\forall x.\varphi}{
\{\;\Gamma\Rightarrow \varphi(t)\;\}_{t\in \Tm}},
\qquad
\infer[(\forall x\Left)]{\forall x.\varphi, \Gamma\Rightarrow\Pi}{
\varphi(t),\forall x.\varphi, \Gamma\Rightarrow\Pi}.
$$
By the induction hypothesis, we have 
$\vdash_m 
\varphi(t), \Gamma\Rightarrow\Pi$.
Hence applying $(\Cut)$ on $\varphi(t)$, we obtain 
$\vdash_m \Gamma\Rightarrow\Pi$,
noting that $\Rk(\varphi(t)) < m$.\\

\noindent
(4) $\pi_l$ and $\pi_r$ respectively end with
$$
\infer[(\forall X\Right)]{\Gamma\Rightarrow
\forall X.\varphi}{\Gamma\Rightarrow \varphi(Y)}
\qquad
\infer[(\Omega_k\Left)]{
\forall X.\varphi, \Gamma \Rightarrow \Pi}{
\{\;\Delta,\forall X.\varphi, \Gamma\Rightarrow\Pi\;\}_{\Delta \in |\forall X.\xi|_k}}
$$
By the induction hypothesis, we have
$\vdash_m 
\Delta, \Gamma\Rightarrow\Pi$ for 
every $\Delta \in |\forall X.\xi|_k$. Hence we may apply
$$
\infer[(\tilde{\Omega}_k\Left)]{
\Gamma\Rightarrow\Pi}{
\Gamma\Rightarrow \varphi(Y) &
\{\;\Delta,\Gamma\Rightarrow\Pi\;\}_{\Delta \in |\forall X.\xi|_k}}
$$
to conclude 
$\vdash_m \Gamma\Rightarrow\Pi$.
\end{proof}

\begin{lem}\label{l-cut1}
Let $m<\omega$. 
Suppose that $\Rk(\varphi) \leq m$, 
$\vdash_m \varphi,\Gamma\Rightarrow\Pi$ and 
$\vdash_m \Gamma\Rightarrow\varphi$.
Then
$\vdash_m \Gamma\Rightarrow \Pi$.
\end{lem}

\begin{proof}
By structural induction on the derivation of 
$\Gamma\Rightarrow\varphi$. 
If the main formula of the last inference is $\varphi$,
it follows from Lemma \ref{l-cut0}.
Suppose otherwise. For instance, suppose that 
the derivation ends with
$$
\infer[(\Omega_k\Left)]{
\forall X.\xi, \Gamma \Rightarrow \varphi}{
\{\;\Delta,\Gamma\Rightarrow\varphi\;\}_{\Delta \in |\forall X.\xi|_k}}.
$$
By the induction hypothesis, we have
$
\vdash_m \Delta, \Gamma\Rightarrow\Pi$ for every 
$\Delta \in |\forall X.\xi|$. Hence
$\vdash_m
\forall X.\xi, \Gamma\Rightarrow\Pi$ by rule 
$(\Omega_k\Left)$.
\end{proof}

\begin{lem}\label{l-cut2}
Let $m < \omega$. If 
$\vdash_{m+1} \Gamma\Rightarrow \Pi$, then 
$\vdash_{m} \Gamma\Rightarrow \Pi$.
\end{lem}

\begin{proof}
By structural induction on the derivation.
Suppose that it ends with
$$
\infer[(\Cut)]{\Gamma\Rightarrow\Pi}{
\Gamma\Rightarrow\varphi & \varphi,\Gamma\Rightarrow\Pi}.
$$
By the induction hypothesis, we have
$\vdash_m \Gamma\Rightarrow\varphi$ and 
$\vdash_m \varphi,\Gamma\Rightarrow\Pi$. Moreover
$\Rk(\varphi)\leq m$. Hence
$\vdash_m \Gamma\Rightarrow\Pi$ by Lemma \ref{l-cut1}.
\end{proof}

\begin{lem}\label{l-cut3}
If
$\vdash_{\omega} \Gamma\Rightarrow \Pi$, then 
$\vdash_{0} \Gamma\Rightarrow \Pi$.
\end{lem}

\begin{proof}
By structural induction on the derivation.
Suppose that it ends with
$$
\infer[(\Cut)]{\Gamma\Rightarrow\Pi}{
\Gamma\Rightarrow\varphi & \varphi,\Gamma\Rightarrow\Pi}.
$$
By the induction hypothesis, we have
$\vdash_0 \Gamma\Rightarrow\varphi$ and 
$\vdash_0 \varphi,\Gamma\Rightarrow\Pi$. 
Let $\Rk(\varphi) = m$, then we have 
$\vdash_{m+1} \Gamma\Rightarrow\Pi$. Hence applying 
Lemma \ref{l-cut2} $m+1$ times, we obtain
$\vdash_{0} \Gamma\Rightarrow\Pi$.
\end{proof}

This completes the proof of Lemma \ref{l-cut}.

\section{Algebraic proof of cut elimination for $\LIT$}
\label{a-okada}

We here outline an algebraic proof of cut elimination 
for the full second order calculus $\LIT$ that we attribute 
to Maehara \cite{Maehara91}
and Okada \cite{Okada96,Okada02}. 
This will be useful for comparison with the parameter-free case $\LIP_{n+1}$,
that we have addressed in the main text.

Let $\fwp (\FM)$ be the set of finite sets of second order formulas, 
so that $\langle \fwp (\FM), \cup, \emptyset\rangle$ is a commutative 
idempotent monoid, 
and $\SEQ$ be
the set of sequents of $\LIT$. As before, we have
$\ldd : \fwp (\FM) \times \SEQ \longrightarrow \SEQ$ 
defined by 
$\Gamma \ldd (\Sigma\Rightarrow \Pi) := (\Gamma,\Sigma \Rightarrow \Pi)$.
So
$$\m{CF} \Def \langle \fwp(\FM), \SEQ, \Rcf, \cup, \emptyset,
\ldd \rangle$$
is a Heyting frame, where $\Gamma \Rcf (\Sigma\Rightarrow\Pi)$
iff $\LIT\vdash^{cf} \Gamma, \Sigma \Rightarrow \Pi$.

Define a Heyting-valued prestructure $\mc{CF} := \langle \m{CF}^+, \Tm, \mc D, 
\mc L\rangle$ 
by 
$p^{\mc{CF}}(\vec{t})  :=  \gamma(p(\vec{t}))$ for each predicate symbol $p$
and 
$$\mc D \Def 
 \{ F \in \Gal(\m{CF})^\Tm : F \mbox{ matches some } \tau \in 
\ABS\},$$
where $F$ \emph{matches} $\tau$ just in case 
$\tau(t) \in F(t) \subseteq \tau(t)^{\lhd}$ holds for every 
$t \in \Tm$.
This choice of $\mc D \subseteq 
\Gal(\m{CF})^\Tm$ is a logical analogue of Girard's \emph{reducibility
candidates} as noticed by Okada.

For instance, given a set variable $X$, define
$F_X \in \Gal(\m{CF})^\Tm$ by $F_X(t) := \gamma(X(t))$.
Then $X(t) \in F_X(t) \subseteq X(t)^\lhd$. Hence 
$F_X$ matches $X = \lambda x.X(x)$, so belongs to $\mc D$.

Given a set substitution $\bullet: \VAR \longrightarrow \ABS$ and a valuation 
$\mc V: \VAR \longrightarrow {\mc D}$,
we say that $\mc V$ \emph{matches} $\bullet$ 
if $\mc V (X)$ matches $X^\bullet \in \ABS$
for every $X\in \VAR$. That is, $X^\bullet(t) \in \mc V(X(t)) \subseteq
X^\bullet (t)^{\lhd}$ holds for every $X\in \VAR$ and $t\in \Tm$.

\begin{lem}\label{l-okada}
Let $\bullet$ be a set substitution
and $\mc V$ a valuation that matches $\bullet$.
Then for every $\varphi \in \FM$, 
$$\varphi^\bullet \in 
\mc V (\varphi) \subseteq \varphi^{\bullet\lhd}.$$
\end{lem}

\begin{proof}
By induction on the structure of $\varphi$.\\

\noindent
(1) $\varphi$ is an atom $X(t)$. By assumption we have
$X^\bullet (t) \in \mc V(X(t)) \subseteq X^\bullet(t)^\lhd$.\\

\noindent
(2) The outermost connective of $\varphi$ is first order.
Similar to the proof of Lemma \ref{l-main}.\\

\noindent
(3) $\varphi = \forall X.\psi(X)$.
We first prove $(\forall X.\psi)^\bullet \in 
\mc V(\forall X.\psi) = \bigcap_{F\in \mc D}\mc V[F/X](\psi)$.
So let $F \in \mc D$ that matches $\tau \in \ABS$.
We update substitution $\bullet$ by letting $X^\bullet := \tau$
so that $\bullet$ matches $\mc V[F/X]$.
By the induction hypothesis $\psi^\bullet(\tau) \in 
\mc V[F/X](\psi)$. Hence for any $(\Gamma\Rightarrow\Pi) \in 
\mc V[F/X](\psi)^\rhd$, we have $\psi^\bullet(\tau),
\Gamma\Rcf\Pi$. So $(\forall X.\psi)^\bullet,
\Gamma\Rcf\Pi$ by rule $(\forall X\Left)$. Therefore
$(\forall X.\psi)^\bullet \in \mc V[F/X](\psi)$.
This proves $(\forall X.\psi)^\bullet \in 
\mc V(\forall X.\psi)$.

We next prove 
$\mc V(\forall X.\psi) \subseteq (\forall X.\psi)^{\bullet\lhd}$.
Let $\Gamma \in \mc V(\forall X.\psi)$. Choose a variable 
$Y$ such that $Y \not\in \FV(\Gamma)$.
By the induction hypothesis,
we have $\Gamma \in \mc V[F_Y/X](\psi)\subseteq 
(\psi^\bullet(Y))^\lhd$, that is, $\Gamma\Rcf\psi^\bullet(Y)$.
Hence $\Gamma\Rcf (\forall X.\psi)^\bullet$ by rule $(\forall \Right)$.
This proves $\Gamma \in (\forall X.\psi)^{\bullet\lhd}$.\\

\noindent
(3) $\varphi = \exists X.\psi(X)$.
We first prove $(\exists X.\psi)^\bullet \in 
\mc V(\exists X.\psi) = \gamma \left(\bigcup_{F \in\mc D} \mc V[F/X](\psi)
\right)$. Let $(\Gamma\Rightarrow\Pi)\in 
\left(\bigcup_{F \in\mc D} \mc V[F/X](\psi)
\right)^\rhd$ and choose a variable $Y$ such that
$Y \not\in \FV(\Gamma,\Pi)$. By the induction hypothesis,
we have $\psi^\bullet(Y) \in \mc V[F_Y/X](\psi) \subseteq 
\bigcup_{F \in\mc D} \mc V[F/X](\psi)$, so
$\psi^\bullet(Y), \Gamma\Rcf\Pi$.
Hence $(\exists X.\psi)^\bullet, \Gamma\Rcf\Pi$ by rule $(\exists X\Left)$.
Therefore $(\exists X.\psi)^\bullet \in \mc V(\exists X.\psi)$.

We next prove $\mc V(\exists X.\psi) \subseteq 
(\exists X.\psi)^{\bullet\lhd}$.
Let $\Gamma \in 
\bigcup_{F \in\mc D} \mc V[F/X](\psi)$, that is,
$\Gamma \in \mc V[F/X](\psi)$ for some $F \in \mc D$ (that match $\tau$).
By the induction hypothesis, 
$\Gamma \in \mc V[F/X](\psi) \subseteq (\psi^\bullet(\tau))^\lhd$,
so $\Gamma \Rcf \psi^\bullet (\tau)$. Hence
$\Gamma \Rcf (\exists X.\psi)^{\bullet}$ by rule $(\exists X\Right)$,
so 
$\Gamma \in (\exists X.\psi)^{\bullet\lhd}$.
This proves 
$\mc V(\exists X.\psi) \subseteq 
(\exists X.\psi)^{\bullet\lhd}$.
\end{proof}

As a consequence:

\begin{lem}
$\mc{CF}$ is a Heyting structure.
\end{lem}

\begin{proof}
Let $\mc V$ be a valuation and $\tau$ an abstract.
Our goal is to show that $\mc V(\tau) \in \mc D$, that is,
there is some $\tau_0$ such that 
$\tau_0(t) \in \mc V(\tau(t)) \subseteq \tau_0(t)^\lhd$
holds for every $t\in \Tm$. Here 
$\mc V(\tau)$ is defined by 
$\mc V(\tau)(t) := \mc V(\tau(t))$.

Since $\mc V$ is a valuation into $\mc D$, every set variable $X$
is associated with an abstract $\tau_X$ so that
$\mc V(X)$ matches $\tau_X$. Define a set substitution $\bullet$
by $X^\bullet := \tau_X$. Then $\mc V$ matches $\bullet$.
Hence by Lemma \ref{l-okada} we obtain
$\tau^{\bullet}(t) \in \mc V(\tau(t)) \subseteq \tau^\bullet(t)^\lhd$.
Thus $\tau^\bullet$ is the desired abstract.
\end{proof}

Next, define a valuation $\mc I$
by $\mc I (X) := F_X$. We then have
$X(t) \in \mc I(X(t)) = \gamma(X(t)) \subseteq X(t)^\lhd$ for every $X\in \VAR$ and 
$t\in \Tm$, so 
$\mc I$ matches the identity substitution. Hence 
we have $\varphi \in 
\mc I (\varphi) \subseteq \varphi^{\lhd}$ for every formula
$\varphi\in \FM$.

\begin{thm}[Algebraic cut elimination for $\LIT$] \label{t-mo}
For every sequent $\Gamma\Rightarrow\Pi$, the following are equivalent.
\begin{enumerate}
\item $\Gamma\Rightarrow\Pi$ is provable in $\LIT$.
\item $\Gamma\Rightarrow\Pi$ is valid in all Heyting-valued structures.
\item $\Gamma\Rightarrow\Pi$ is cut-free provable in $\LIT$.
\end{enumerate}
\end{thm}

\begin{proof}
((1) $\Rightarrow$ (2)) By Lemma \ref{l-sound}.\\
((2) $\Rightarrow$ (3)) By
$\Gamma \in \mc I (\Gamma) \subseteq \mc I (\Pi) \subseteq 
\Pi^{\lhd}$.\\
((3) $\Rightarrow$ (1)) Trivial.
\end{proof}

\end{document}